\documentclass[12pt]{article}
\usepackage{amsmath, amsfonts, amssymb, amsthm, graphicx, geometry, arydshln, umoline, subfig, newtxtext, newtxmath, xcolor, comment, soul, enumerate, bm, tikz, float}
\usepackage[colorlinks=true,citecolor=blue]{hyperref}
\usepackage{natbib}
\usepackage[toc,title,page]{appendix}
\usetikzlibrary{positioning,calc}

\usepackage{booktabs}

\theoremstyle{definition}
\newtheorem{theorem}{Theorem}
\newtheorem*{theorem*}{Theorem}

\newtheorem{claim}{Claim}

\newtheorem{corollary}{Corollary}

\newtheorem{example}{Example}

\newtheorem{lemma}{Lemma}

\newtheorem{remark}{Remark}

\newcommand{\mR}{\mathbb{R}}

\begin{document}

\title{Characterizing the ELS Values \\ with Fixed-Population Invariance Axioms\thanks{We thank We thank the Associate Editor and the two anonymous referees for their constructive comments, which have been helpful in improving the paper. We also thank
seminar participants at Tokyo social choice theory workshop, Waseda University, Hitotsubashi Summer Institute 2025, and I France-Spain Meeting on Microeconomic Theory at University of Granada for their helpful comments. Funaki and Nakada acknowledge  the financial supports from Japan Society for the Promotion of Science KAKENHI: No. 22H00829 (Funkai), No.25K16606 (Nakada).
Koriyama and Tamura acknowledge the financial support from Investissements d’Avenir, ANR-11-IDEX-0003/Labex Ecodec/ANR-11-LABX-0047.}
}

\author{
Yukihiko Funaki\footnote{School of Political Science and Economics, Waseda University. E-mail:\texttt{funaki@waseda.jp}.}
\and
Yukio Koriyama\footnote{CREST, Ecole Polytechnique, Institut Polytechnique de Paris. E-mail:\texttt{yukio.koriyama@polytechnique.edu}.}
\and
Satoshi Nakada\footnote{School of Management, Department of Business Economics, Tokyo University of Science. E-mail:\texttt{snakada@rs.tus.ac.jp}.}
\and Yuki Tamura\footnote{CREST, Ecole Polytechnique, Institut Polytechnique de Paris. E-mail:\texttt{yuki.tamura@polytechnique.edu}.}
}
\date{\today}
\maketitle

\begin{abstract}
We study efficient, linear, and symmetric (ELS) values, a central family of allocation rules for cooperative games with transferable-utility (TU-games) that includes the Shapley value, the CIS value, and the ENSC value. We first show that every ELS value can be written as the Shapley value of a suitably transformed TU-game. We then introduce three types of invariance axioms for fixed player populations. The first type consists of composition axioms, and the second type is active-player consistency. Each of these two types yields a characterization of a subclass of the ELS values that contains the family of least-square values. Finally, the third type is nullified-game consistency: we define three such axioms, and each axiom yields a characterization of one of the Shapley, CIS, and ENSC values.
\\
\end{abstract}

\newpage

\section{Introduction}\label{sec_intro}

Designing allocation rules that satisfy normatively appealing desiderata is a central theme in cooperative games with transferable utility (TU-games). Since \citet{Shapley1953}, a large axiomatic literature has evaluated solution concepts and proposed axioms, including many forms of \textit{invariance}. In this paper, we focus on an idea that underlies several invariance axioms. Start from a problem and an allocation selected by some rule. Fix certain components of this allocation, and construct a reduced problem by adjusting the original one to account for those fixed components. The rule is invariant if, in the reduced problem, it assigns the relevant agents exactly the same payoffs as in the original problem.

\emph{Consistency} is a prime example of such an invariance axiom.\footnote{For a comprehensive survey of consistency, see \citet{thomson2012axiomatics}.} Starting from a problem and an allocation chosen by the rule, suppose a subset of agents leaves while carrying the components assigned to them. We then construct the corresponding reduced problem for the agents who remain. The rule is consistent if, in this reduced problem, it assigns to the remaining agents exactly the same payoffs as in the original problem. This formulation captures the idea that when some players depart with their allocated shares, the assignments to those who remain should be unaffected. Consistency is typically formulated in \emph{variable-population} frameworks, which model situations in which players may exit and the problem is thereby reduced.

In many applications, however, the set of players is not meant to vary. In intra-firm negotiations, resource allocation within a community, or cost sharing within a fixed group, the relevant group of agents is given and does not change. In such \emph{fixed-population} frameworks, it is natural to ask which allocation rules satisfy analogous invariance requirements when these requirements are formulated without allowing players to leave.

This paper studies three types of fixed-population invariance axioms and identifies which efficient, linear, and symmetric allocation rules (ELS values) in TU-games satisfy them. Efficiency, linearity, and symmetry are widely regarded as foundational requirements for allocation rules. The ELS class is broad, encompassing well-known rules such as the Shapley value \citep{Shapley1953}, the egalitarian value, the center-of-imputation-set (CIS) value, and its dual, the egalitarian non-separable contribution (ENSC) value \citep{DriessenFunaki1991}. 

The first type consists of \emph{composition} axioms, which originate in the bankruptcy and surplus-sharing literature \citep{young1988jet,moulin1987ijgt}.\footnote{For a comprehensive survey of the bankruptcy problem, including results on composition axioms, see \citet{thomson2019book}.} In our companion paper \citep{fknt2025wppd}, we introduced TU-game analogues of \emph{composition up} and \emph{composition down}. When moving from bankruptcy and surplus-sharing problems to TU-games, however, the construction of the corresponding game is no longer unique. For this reason, in the present paper we distinguish two versions of composition down, namely \emph{insider-guaranteed composition down} and \emph{outsider-guaranteed composition down}. The companion paper focuses on the insider-guaranteed version, whereas the outsider-guaranteed version is introduced here. Informally, composition axioms require that a large allocation problem can be decomposed into smaller subproblems and then recomposed without changing the final outcome.

Composition up proceeds in two stages: first, choose any provisional worth for the grand coalition and allocate as if this provisional worth were the true grand-coalition worth; second, adjust by reducing each coalition’s worth by what its members received in the first stage and allocate the remainder (or cover any shortfall) so that the totals match the actual grand-coalition worth. Composition up requires that the sum of these two stages coincide with the allocation obtained by applying the allocation rule directly to the original game.

Composition down likewise yields the allocation in two stages without changing the outcome. First, choose any provisional worth for the grand coalition and compute reference payoffs as if this provisional worth were the true grand-coalition worth. Then, keeping the grand-coalition worth at its actual level, transform the game by applying one of the following operations to every proper coalition: (a) replace each coalition’s worth by the sum of its members’ reference payoffs (the \emph{insider-guaranteed} version); or (b) replace each coalition’s worth by the grand-coalition worth minus the sum of the reference payoffs of players outside the coalition (the \emph{outsider-guaranteed} version). Composition down requires that the allocation rule applied to either transformed game coincide with the allocation rule applied directly to the original game.

Within the class of ELS values, we prove that composition up, insider-guaranteed composition down, and outsider-guaranteed composition down are equivalent, and we characterize the ELS values that satisfy these equivalent axioms: they are either the egalitarian value or an affine combination of allocation rules introduced by \citet{dragan1992average} (Theorem \ref{comp ELS}).

The second axiom, \emph{active-player consistency}, considers situations in which a coalition of players actively seeks cooperation from players outside the coalition. Suppose such a coalition negotiates with these external players, who agree to cooperate on the condition that their payoffs are determined by the underlying rule. After this agreement is fixed, the remaining worth is distributed among the members of the active coalition. An allocation rule satisfies active-player consistency if each active player receives exactly the same payoff as in the original game, regardless of the intermediate bargaining with outsiders. We establish a characterization of the ELS values that satisfy active-player consistency (Theorem \ref{consistency ELS}). The resulting class coincides with the affine combinations of allocation rules identified by the composition axioms. We also show that this subfamily is closely related to the \emph{least-square values} of \cite{ruiz1998family}, which are efficient allocations that minimize a weighted variance of coalition excesses (Corollary \ref{cor:leastsquare}). Thus, our result provides a  axiomatic characterization of least-square values via active-player consistency.

The third type of properties are \emph{nullified-game consistency} axioms, introduced for set-valued solutions by \citet{kn2025el} to characterize the core in TU-games. It addresses situations in which some players’ payoffs are fixed (e.g., by binding agreements or commitments). The requirement is that, when such payoffs are fixed and the contributions of those players are neutralized, the solution for the remaining players, computed in the corresponding reduced game, coincides with their payoffs in the original game. Importantly, the player set remains unchanged; the \emph{nullified} players simply no longer affect the residual allocation problem. We study three variants, paralleling the constructions used to define variable-population consistency axioms in \citet{hart1989potential}, \citet{funaki1996wp}, and \citet{moulin1985separability}, and show that, within the ELS values, they respectively single out the Shapley, CIS, and ENSC values (Theorem \ref{NGC axiomatization}).

The remainder of the paper is organized as follows. Section \ref{sec_preliminary} introduces the model. Section \ref{sec:els} defines the ELS values and presents a representation via Shapley values of suitably modified games, yielding a characterization of all ELS values. Section \ref{sec_consistency} introduces the composition axioms and active-player consistency, and establishes a characterization result based on these axioms. Section \ref{sec_nullified} analyzes nullified-game consistency and provides characterizations of the Shapley, CIS, and ENSC values. Section \ref{sec_conclude} concludes. Omitted proofs are collected in the appendices.

\section{Preliminaries}\label{sec_preliminary}

Let $N=\{1,\ldots, n\}$ with $n \ge 2$ be the set of players. A cooperative game with transferable utility (TU-game) is a pair $(N,v)$ where the characteristic function $v:2^N\to\mathbb{R}$ satisfies $v(\emptyset)=0$.  
Throughout, we fix $N$ and identify a game with its characteristic function $v$. Let $\mathcal{V}$ denote the set of all TU-games on~$N$. A solution (or allocation rule) is a function $\varphi: \mathcal{V} \to \mR^n$ that assigns to each game $v \in \mathcal{V}$ a payoff vector $\varphi(v) = (\varphi_i(v))_{i \in N}$. 

A coalition is any nonempty subset of $N$.  For $v \in \mathcal{V}$ and $i\in N$, player~$i$ is a \textbf{null player in $\bm{v}$} if $v(S\cup\{i\})=v(S)$ for all $S\subseteq N\setminus\{i\}$. 
Let $\mathrm{Null}(v)$ denote the set of null players in $v$. Let $\Pi$ be the set of permutations of $N$, and write $\pi S:=\{\pi(i):i\in S\}$.

We define four classes of games used throughout the paper. For $\emptyset \neq T\subseteq N$, the \textbf{$\bm{T}$-unanimity game} is defined by
\begin{equation*}
u_T(S)=
\begin{cases}
1 & \text{if } T\subseteq S,\\
0 & \text{otherwise}. 
\end{cases}
\end{equation*}
For $v\in\mathcal{V}$, its \textbf{dual} $v^*\in\mathcal{V}$ is defined by
\begin{equation*}
v^*(S)=v(N)-v(N\setminus S) \ \text{for all} \ S\subseteq N.
\end{equation*}
For $x\in\mathbb{R}^n$, we write $\overline{x}=\frac{1}{n}\sum_{i\in N}x_i$.
Also, the \textbf{additive game induced by $\bm{x}$}, denoted $\widehat{x}\in\mathcal{V}$, is defined by
\begin{equation*}
\widehat{x}(S)=\sum_{i\in S}x_i \ \text{for all} \ \emptyset \neq S\subseteq N.
\end{equation*}
For $\pi \in \Pi$ and $v \in \mathcal{V}$, the \textbf{$\bm{\pi}$-permutation of $\bm{v}$}, denoted $\pi v\in\mathcal{V}$, is defined by
\begin{equation*}
\pi v(\pi S)=v(S) \ \text{for all} \ S\subseteq N.
\end{equation*}


\section{ELS values}\label{sec:els}
\subsection{Definitions}
We study solutions satisfying the following three standard axioms.

\medskip

\noindent\textbf{Efficiency (E)}: 
For any $v \in \mathcal{V}$, $\sum_{i \in N}\varphi_i(v)=v(N)$.

\noindent\textbf{Linearity (L)}: 
For any $v, w \in  \mathcal{V}$ and $c, c'\in \mR$, $\varphi(cv+c'w)= c\varphi(v)+c'\varphi(w)$.



\noindent\textbf{Symmetry (SYM)}: 
For any $\pi \in \Pi$, $v \in \mathcal{V}$, and $i \in N$, $\varphi_i(v)=\varphi_{\pi(i)}(\pi v)$.\footnote{This axiom is closely related to \emph{equal treatment of equals}: for any $v\in\mathcal{V}^N$ and $i,j\in N$, if $i$ and $j$ are symmetric in $v$, then $\varphi_i(v)=\varphi_j(v)$. In general, (SYM) is stronger than equal treatment of equals. However, for linear and efficient solutions, the two properties are equivalent; see \citet[Theorem 2]{malawski2007note}.}

\medskip

Following convention, we refer to any solution satisfying (E), (L), and (SYM) as an ELS value.
\citet{ruiz1998family} show that any ELS value can be written as an affine combination of the following solutions introduced by \citet{dragan1992average}. 
For $s\in\{1,\dots,n-1\}$, $v\in\mathcal{V}$ and $i\in N$,
\begin{eqnarray*}
\psi^s_i(v)=\frac{v(N)}{n}+\frac{n-1}{s}
\left(
\frac{\sum_{S \subseteq N: |S|=s}v(S)}{\binom{n}{s}}-\frac{\sum_{S \subseteq N: |S|=s, i \notin S}v(S)}{\binom{n-1}{s}} 
\right),    
\end{eqnarray*}
and for $s = n$, 
\begin{equation*}
\psi^n_i(v)=\frac{v(N)}{n} = ED_i(v).
\end{equation*}

In particular, $\psi^{1}$ coincides with the CIS value and $\psi^{\,n-1}$ coincides with its dual, the ENSC value. Both were introduced by \citet{DriessenFunaki1991} and are defined as follows: for any $v \in \mathcal{V}$ and $i \in N$,
\begin{equation*}
CIS_i(v)=v(\{i\})+\frac{1}{n}\!\left( v(N)-\sum_{k \in N}v(\{k\}) \right),
\end{equation*}
and
\begin{equation*}
ENSC_i(v)=CIS_i(v^*)
= v(N)-v(N \setminus \{i\})+\frac{1}{n}\!\left( v(N)-\sum_{k \in N}\left(v(N)-v(N \setminus \{k\}) \right) \right).
\end{equation*}

Moreover, \citet{dragan1992average} shows that 
$\varphi(v)=\frac{1}{n-1}\sum_{s=1}^{n-1}\psi^s(v)$
coincides with the Shapley value \citep{Shapley1953}: for any $v \in \mathcal{V}$ and $i \in N$,
\begin{equation*}
Sh_i(v)=\sum_{S \subseteq N: i \notin S}\frac{|S|!(n-|S|-1)!}{n!}\left(v(S \cup \{i\})-v(S)\right).
\end{equation*}

For later use, we state a well-known representation of the ELS values. For derivations of the lemma, see \cite{weber1988}, \citet{ruiz1998family}, \citet{Nakada2024}, or \citet{fk2025ijgt}.

\begin{lemma}\label{Linear}
\begin{itemize}\itemsep-0.1cm
\item[]
\item[(i)] $\varphi$ satisfies (L) if and only if, for any $i \in N$, there are constants $(p_i(S))_{S \subseteq N} \in \mathbb{R}^{2^n-1}$ such that for any $v \in \mathcal{V}$,
\[
\varphi_i(v)=\sum_{S \subseteq N}p_i(S)v(S).
\]
\item[(ii)] $\varphi$ satisfies (L) and (SYM) if and only if there are constants $(p_k)_{k=1}^n \in \mathbb{R}^n$ and $(q_k)_{k=1}^{n-1} \in \mathbb{R}^{n-1}$ such that for any $v \in \mathcal{V}$ and $i \in N$,
\[
\varphi_i(v)=\sum_{S \subseteq N: |S| = s,\hspace{0.05cm} i \in S}   p_sv(S)+\sum_{S \subseteq N: |S| = s,\hspace{0.05cm} i \notin S}   q_sv(S).
\]
\item[(iii)] $\varphi$ is an ELS value if and only if the coefficients in (ii) satisfy $p_n=\frac{1}{n}$ and $q_k=-\frac{k}{n-k}p_k$ for all $k=1,\ldots, n-1$.
\end{itemize}
\end{lemma}

\subsection{ELS Values as Shapley Variants}\label{sec_sh}

Among the ELS values, the Shapley value is the best known example. More generally, ELS values can be viewed as variants of the Shapley value. We briefly recall this connection, which is essentially contained in \citet[Proposition 2]{radzikdriessen2013mss}.

Following \citet{yokote2016new}, for $v \in \mathcal{V}$ and $\sigma:\{1,2,\ldots,n\}\to\mathbb{R}$, define $v_\sigma \in \mathcal{V}$ by $v_\sigma(S):=\sigma(|S|)v(S)$
for all $\emptyset \neq S\subseteq N$. The associated \emph{$\sigma$-Shapley value} is defined by $\sigma\text{-}Sh(v):=Sh(v_\sigma)$.
By construction, every $\sigma$-Shapley value satisfies (L) and (SYM), and if $\sigma(n)=1$, then it also satisfies (E).

This class encompasses several variants of the Shapley value studied in the literature. For example, if $\sigma(s)=\delta^{\,n-s}$ for some $\delta\in[0,1]$, then $\sigma\text{-}Sh$ coincides with the $\delta$-discounted Shapley value \citep{Joosten1996thesis,driessen2002weighted}. \citet{yokote2018balanced} introduce the $\mathbf{r}$-egalitarian Shapley values in a variable-population framework. Although their setting differs from ours, the $\sigma$-Shapley values with $\sigma(n)=1$ are closely analogous to that class.

As an immediate consequence of Lemma~\ref{Linear}, the class of $\sigma$-Shapley values is the subclass of solutions satisfying (L) and (SYM) described below.
\begin{lemma}\label{SYML sigma-sh}
A solution $\varphi$ satisfies (L) and (SYM), and its coefficients in the representation of Lemma~\ref{Linear}(ii) satisfy $q_k=-\frac{k}{n-k}p_k$ for all $k=1,\ldots,n-1$, if and only if there exists $\sigma:\{1,\ldots,n\}\to\mathbb{R}$ such that $\varphi=\sigma\text{-}Sh$.
\end{lemma}

We provide a proof of Lemma~\ref{SYML sigma-sh} in Appendix \ref{appendix: sec_sh}. 
Combining Lemmas~\ref{Linear} and~\ref{SYML sigma-sh}, we obtain the following characterization, established by \citet[Proposition 2]{radzikdriessen2013mss}.

\begin{theorem*}[\citealp{radzikdriessen2013mss}, Proposition 2]
A solution $\varphi$ is an ELS value if and only if there exists $\sigma:\{1,2,\ldots,n\}\to\mathbb{R}$ with $\sigma(n)=1$ such that $\varphi=\sigma\text{-}Sh$.
\end{theorem*}

\section{Composition Axioms and Active-Player Consistency}\label{sec_consistency}
In this section, we characterize a subclass of the ELS values using two types of invariance axioms: (i) three \emph{composition} axioms and (ii) \emph{active-player consistency}.

\subsection{Composition Axioms}\label{sec:composition}
In cooperative games, the worth of the grand coalition can change even after an initial allocation, for example, when available resources expand or contract. Composition axioms take that initial allocation as a reference and decompose the game into a baseline component and a residual component; operationally, this yields the same final payoffs as reallocating on the updated grand coalition worth. These axioms were first developed for the bankruptcy problem and its generalization, the surplus-sharing problem \citep{young1988jet, moulin1987ijgt, moulin2000econometrica}, and were later extended to general TU-games by \citet{fknt2025wppd}.

We begin by defining a key operation on games that underpins the composition axioms. For $v \in \mathcal{V}$ and $t \in \mathbb{R}$, let 
\begin{equation*}
{v^t}(S) :=
\begin{cases}
t &\text{if} \ S=N, \\ 
v(S)  &\text{otherwise}.
\end{cases}
\end{equation*}

Our first composition axiom, \emph{composition up} (CU), states that the grand coalition worth can be allocated in two stages. First, choose any provisional grand coalition worth and allocate as if it were the true grand coalition worth. Then adjust by reducing each coalition's worth by what its members received in the first stage and allocate the remainder (or cover any shortfall) so that the totals match the actual grand coalition worth. (CU) requires that the sum of these two stages coincide with the allocation obtained by applying the solution directly to the original game.

\medskip

\noindent\textbf{Composition Up 
(CU)}:\footnote{In the bankruptcy problem, by its nature, one assumes $0 \leq t \le v(N)$. In TU-games, however, no such restriction is required.} 
For any $v\in \mathcal{V}$ and $t \in \mathbb{R}$,
$\varphi (v) = \varphi (v^t)+ \varphi(U(\varphi(v^t),v)) $, where
$$
U(x,v)(S) := v(S) - \sum_{i \in S} x_i,
$$
for all $S \subseteq N$.


The \emph{composition down} (CD) axioms state that the allocation can be obtained in two stages without changing the outcome. First, choose any provisional grand coalition worth and compute the reference payoffs as if it were the true grand coalition worth. Then, keeping the grand coalition worth at its actual worth, transform the game by applying one of the following operations to every proper coalition: (a) replace each coalition's worth by the sum of its members' reference payoffs; or (b) replace each coalition's worth by the grand coalition worth minus the sum of the reference payoffs of players outside the coalition. (CD) requires that the solution applied to either transformed game coincide with the solution applied directly to the original game. We refer to the version based on operation (a) as \emph{Insider-Guaranteed Composition Down} (CD$_{\text{I}}$) and the version based on operation (b) as \emph{Outsider-Guaranteed Composition Down} (CD$_{\text{O}}$).\footnote{The Composition Down introduced in \citet{fknt2025wppd} is  (CD$_{\text{I}}$).}


\medskip

\noindent\textbf{Insider-Guaranteed Composition Down (CD$_{\text{I}}$)
}:\footnote{In the bankruptcy problem, by its nature, one assumes $0 \leq v(N) \leq t$. In TU-games, however, no such restriction is required.}  For any $v\in \mathcal{V}$ and $t \in \mathbb{R}$,
$\varphi (v) =\varphi \left(D_I(\varphi(v^t),v)\right)$, where
$$
D_I(x,v)(S) := 
\begin{cases}
v(N) \ &\text{if} \ S = N \\
\sum_{i \in S} x_i \ &\text{if} \ S \subsetneq N.
\end{cases}
$$

\noindent\textbf{Outsider-Guaranteed Composition Down (CD$_{\text{O}}$)
}: For any $v\in \mathcal{V}$ and $t \in \mathbb{R}$,
$\varphi (v) =\varphi \left(D_O(\varphi(v^t),v)\right)$, where
$$
D_O(x,v)(S) := 
\begin{cases}
v(N) - \sum_{i \notin S} x_i \ &\text{if} \ S \neq \emptyset\\
0 \ &\text{if} \ S = \emptyset.
\end{cases}
$$

Since TU-games provide a richer framework than the bankruptcy problem, there are multiple ways to define transformed games. 
Our two transformation operations, applied after fixing the reference payoffs, are closely related to the standard TU-game extension of a bankruptcy problem and its dual, except that we omit the usual $\min$ and $\max$ truncations. 

We establish the equivalence of (CU), (CD$_{\text{O}}$), and (CD$_{\text{I}}$) within the class of ELS values, and we characterize the ELS values that satisfy these composition axioms.\footnote{Theorem~\ref{comp ELS} continues to hold under assumptions analogous to those in the bankruptcy problem: $0\le t\le v(N)$ for (CU), and $0\le v(N)\le t$ for (CD$_{\mathrm{O}}$) and (CD$_{\mathrm{I}}$).}

\begin{theorem}\label{comp ELS}
Suppose that $\varphi$ is an ELS value. Then the following are equivalent:

\vspace{-0.3cm}

\begin{itemize}\itemsep-0.1cm
\item[(I)] $\varphi$ satisfies (CU).
\item[(II)] $\varphi$ satisfies (CD$_{\text{O}}$).
\item[(III)] $\varphi$ satisfies (CD$_{\text{I}}$).
\item[(IV)] $\varphi$ is either  an affine combination of $(\psi^s)_{s=1}^{n-1}$, or $\varphi=ED$.
\end{itemize}
\end{theorem}

We introduce two axioms used to prove Theorem \ref{comp ELS}; the proof also relies on two lemmas. The proofs of these lemmas are collected in Appendix \ref{appendix:sec_consistency}.

\medskip

\noindent\textbf{Inessential Game Property (IGP)}:\footnote{This axiom was proposed by \cite{rvz1996ijgt}.} For any $x \in \mathbb{R}^n$ and $i \in N$, $\varphi_i(\widehat{x})=x_i$.

\medskip

\noindent\textbf{Renegotiation-proofness (RNP)}: For any $v \in \mathcal{V}$ and $i \in N$, $\varphi_i(\widehat{\varphi(v)})=\varphi_i(v)$.

\medskip

It is straightforward to verify that (IGP) implies (RNP). A canonical example of a solution that satisfies (RNP) but not (IGP) is the egalitarian value.

\begin{lemma}\label{IGP_s}
For any $s=1, \ldots, n-1$, $\psi^{s}$ satisfies (IGP).
\end{lemma}

Let $x \in \mathbb{R}^n$.
Suppose that $\varphi$ is an ELS value. Then,   $\varphi_i(\widehat{x})=\alpha_n ED_i(\widehat{x})+\sum_{s=1}^{n-1}\alpha_s\psi^{s}_i(\widehat{x})$ with $\sum_{s=1}^n\alpha_s=1$. 
By Lemma~\ref{IGP_s},
\begin{equation}
\varphi_i(\widehat{x})
=\alpha_n \overline{x}+ \left(\sum_{s=1}^{n-1}\alpha_s\right)  x_i
=\alpha_n \overline{x}+(1-\alpha_n)x_i.
\label{ELS_inessnetial}
\end{equation}

\begin{lemma}\label{IGP_cons}
The following implications hold. 

\vspace{-0.3cm}

\begin{itemize}\itemsep-0.1cm
\item[(i)] (L) and (CU) together imply (RNP).
\item[(ii)] (E) and (CD$_{\text{O}}$) together imply (RNP).
\item[(iii)] (E) and (CD$_{\text{I}}$) together imply (RNP).
\end{itemize}
\end{lemma}

\begin{proof}[Proof of Theorem \ref{comp ELS}.]
(IV) $\Rightarrow$ (I), (II) and (III): It is straightforward to show that  $\varphi=ED$ satisfies (CU), (CD$_{\text{O}}$), and (CD$_{\text{I}}$).
Let $\varphi$ be an affine combination of $(\psi^s)_{s=1}^{n-1}$.
Then, by letting $\alpha_n=0$ in \eqref{ELS_inessnetial}, $\varphi$ satisfies
\begin{equation}
    \varphi(\widehat{x})=x    
    \label{eq:x_i}
\end{equation}
for any $x\in \mathbb{R}^n$.
We show that $\varphi$ satisfies all the composition axioms.

\medskip

By definition of $U$, we have $U(\varphi(v^t),v)=v-\widehat{\varphi(v^t)}$. 
By \eqref{eq:x_i} and (L), 
$
\varphi(v^t)+\varphi(U(\varphi(v^t),v))
=\varphi(v^t)+\varphi(v)-\varphi(v^t)=\varphi(v),
$
implying that $\varphi$ satisfies (CU).

Given $v \in \mathcal{V}$ and $t \in \mathbb{R}$, define the game $w \in \mathcal{V}$ by $w(S)=v(N)-t$ for any $\emptyset \neq S \subseteq N$. 
By definition of $D_O$, we have $D_O(\varphi(v^t),v)=\widehat{\varphi(v^t)}+w$.
By \eqref{eq:x_i} and (L), 
\[
\varphi\left(D_O(\varphi(v^t),v)\right)
=\varphi(\widehat{\varphi(v^t)})+\varphi(w)
=\varphi(v^t)+\varphi(w)
=\varphi(v^t+w).
\]
By definition of $w$, $v^t+w=v+w^0$.
Since $\varphi(w^0)=0$ by (E) and (SYM), 
$\varphi(v^t+w)
=\varphi(v+w^0)
=\varphi(v)+\varphi(w^0)
=\varphi(v)$, implying that $\varphi$ satisfies (CD$_{\text{O}}$).

By definition of $D_I$, we have $D_I(\varphi(v^t),v)=\widehat{\varphi(v^t)}+(v(N)-t)u_N$.
By \eqref{eq:x_i} and (L),
\[
\varphi(D_I(\varphi(v^t),v)
=\varphi(\widehat{\varphi(v^t)})+\varphi((v(N)-t)u_N)
=\varphi(v^t)+\varphi((v(N)-t)u_N)
=\varphi(v^t+(v(N)-t)u_N).
\]
Since $v^t+(v(N)-t)u_N=v$ by definition of $v^t$, $\varphi$ satisfies (CD$_{\text{I}}$).

\medskip

\noindent (I), (II) or (III) $\Rightarrow$ (IV): By Lemma \ref{IGP_cons}, an ELS value that satisfies one of the composition axioms also satisfies (RNP). Hence, by \eqref{ELS_inessnetial},
$
\varphi_i(v)
=\varphi_i(\widehat{\varphi(v)})
=\alpha_n \overline{\varphi(v)}+(1-\alpha_n)\varphi_i(v)$, which is equivalent to 
\[
\alpha_n\left(\varphi_i(v)-\overline{\varphi(v)}\right)=0
\]
for all $v \in \mathcal{V}$ and $i \in N$. Therefore, either (a) $\alpha_n=0$, or (b) $\varphi_i(v)=\overline{\varphi(v)}$ for all $v \in \mathcal{V}$ and $i \in N$. In case (a), $\varphi$ is an affine combination of $(\psi^s)_{s=1}^{n-1}$; in case (b), $\varphi$ coincides with the ED value. 
\end{proof}

\medskip

In Appendix \ref{appendix:independence}, we provide examples showing that (i) the three composition axioms are independent of one another, and (ii) if any one of (E), (L), or (SYM) is dropped, there exists a solution that is neither ED nor an affine combination of $(\psi^s)_{s=1}^{n-1}$.

\subsection{Active-Player Consistency}
We now consider a situation in which a coalition of players $S$ takes the initiative and seeks cooperation from the players outside the coalition. The outsiders agree to participate, but only on the condition that their payoffs are fixed ex ante: each outsider $j$ insists on receiving exactly the payoff that the solution assigns to $j$ in the original game. Once these payoffs to outsiders are guaranteed, the only remaining worth to be allocated is whatever is left after covering those guarantees. Equivalently, we obtain a reduced game in which, for every coalition $T$, its effective worth is the original worth minus the total guaranteed payoffs to the outsiders in $T$. \emph{Active-player consistency} (AC) asks that each member of the initiating coalition $S$ should receive the same payoff in this reduced game as in the original game.

(AC) can be interpreted as an internal consistency requirement. The solution prescribes a feasible payoff vector for the grand coalition in the original game. The requirement is that any coalition $S$ should be able to implement its part of this outcome by honoring the outsiders' prescribed payoffs and reallocating only the residual worth, without having to revise the payoffs of its own members. In other words, taking the lead in organizing cooperation should not force the active players to deviate from the payoffs assigned to them by the solution in the original game.

Formally, given $S \subsetneq N$, $v \in \mathcal{V}$ and $x \in \mathbb{R}$, let $R^{\text{AC},S}: \mathbb{R}^n \times \mathcal{V} \to \mathcal{V}$ be defined by
$$
R^{\text{AC},S}(x,v)(T) := v(T) - \sum_{i \in T \setminus S} x_i.
$$
for all $T \subseteq N$.\footnote{As an illustration of the reduced game construction, suppose that $n=2$ and $S=\{1\}$. Then $R^{\text{AC},S}(x,v)(\{1\}) = v(\{1\}) - \sum_{i \in \{1\} \setminus \{1\}} x_i = v(\{1\}) - 0$, $R^{\text{AC},S}(x,v)(\{2\}) = v(\{2\}) - \sum_{i \in \{2\} \setminus \{1\}} x_i = v(\{2\}) - x_2$, and $R^{\text{AC},S}(x,v)(N) = v(N) - \sum_{i \in N \setminus \{1\}} x_i = v(N) - x_2$.}

\medskip

\noindent\textbf{\textbf{Active-Player Consistency} (AC)}: For any $v \in \mathcal{V}$, $S \subsetneq N$, and $i \in S$, $\varphi_i(v)=\varphi_i\left(R^{\text{AC},S}\left(\varphi(v),v\right)\right)$.

\medskip

Using (AC), we obtain a characterization of a subclass of the ELS values.

\begin{theorem}\label{consistency ELS}
Suppose that $\varphi$ is an ELS value. Then the following are equivalent:

\vspace{-0.3cm}

\begin{itemize}\itemsep-0.1cm
\item[(I)] $\varphi$ satisfies (AC).
\item[(II)] $\varphi$ is an affine combination of $(\psi^s)_{s=1}^{n-1}$.
\end{itemize}
\end{theorem}
\begin{proof}
Let $\varphi$ be an ELS value: $\varphi_i(v)=\alpha_n ED_i(v)+\sum_{s=1}^{n-1}\alpha_s\psi^{s}_i(v)$ with $\sum_{s=1}^n\alpha_s=1$.
For any $S \subsetneq N$, define $\varphi^S$ as follows: for any $v \in \mathcal{V}$, 
$$
\varphi_i^S(v)=
\begin{cases}
\varphi_i(v) &\text{if} \ i \notin S, \\
0 &\text{if} \ i \in S.
\end{cases}
$$
Then, $R^{\text{AC},S}\left(\varphi(v),v\right)=v- \widehat{\varphi^S(v)}$. Moreover, by \eqref{ELS_inessnetial}, for any $i\in S$, because $\varphi_i^S(v)=0$, we have $\varphi_i(\widehat{\varphi^S(v)})=\alpha_n \overline{\varphi^S(v)}$. Hence, by (L),
(AC) is equivalent to 
\[
\varphi_i(v)=\varphi_i(v- \widehat{\varphi^S(v)}) 
\Leftrightarrow
\varphi_i(\widehat{\varphi^S(v)})
=0
\Leftrightarrow
\alpha_n \overline{\varphi^S(v)}=0
\]
for any $v \in \mathcal{V}$, $S \subsetneq N$, and $i \in S$.
This holds if and only if $\alpha_n=0$, that is, $\varphi$ is an affine combination of $(\psi^s)_{s=1}^{n-1}$.
\end{proof}

\medskip

\begin{remark}
Although the axiom (AC) does not explicitly impose that 
$
\varphi_{i}(R^{AC,S}(\varphi(v),v)) = 0
$
for all $ i \in N \setminus S $, this property follows as a consequence of (AC) when $ \varphi $ is an ELS value. To see this, recall from the above proof that
$
R^{AC,S}(\varphi(v),v) = v - \widehat{\varphi^{S}(v)}.
$
By linearity of $ \varphi $, for each $ i \in N \setminus S $,
\[
\varphi_{i}(R^{AC,S}(\varphi(v),v)) = \varphi_{i}(v) - \varphi_{i}(\widehat{\varphi^{S}(v)}).
\]
By (\ref{ELS_inessnetial}),
$
\varphi_{i}(\widehat{\varphi^{S}(v)}) = \alpha_{n}\overline{\varphi^{S}(v)} + (1 - \alpha_{n})\varphi_{i}^{S}(v).
$
By definition of $ \varphi^{S} $, we have $ \varphi_{i}^{S}(v) = \varphi_{i}(v) $ for each $ i \in N \setminus S $. Substituting yields
\[
\varphi_{i}(R^{AC,S}(\varphi(v),v)) 
= \varphi_{i}(v) - \bigl(\alpha_{n}\overline{\varphi^{S}(v)} + (1 - \alpha_{n})\varphi_{i}(v)\bigr)
= \alpha_{n}\bigl(\varphi_{i}(v) - \overline{\varphi^{S}(v)}\bigr).
\]
By Theorem \ref{consistency ELS}, (AC) implies $ \alpha_{n} = 0 $. It follows that
$
\varphi_{i}(R^{AC,S}(\varphi(v),v)) = 0
$
for each $ i \in N \setminus S $. \hfill $\blacksquare$
\end{remark}

Theorem~\ref{consistency ELS}, we study the implications of (AC) within the class of ELS values. However, the full strength of (L) is not needed for that characterization. We refer the reader to Appendix~\ref{appendix:discussion_ac} for further discussion.

We provide examples showing that the axioms used in Theorem~\ref{consistency ELS} are independent. First, if (E) is dropped, consider the solution defined by $\varphi_i(v)=0$ for every $v\in\mathcal{V}$ and every $i\in N$. This solution satisfies (L), (SYM), and (AC), but violates (E). Second, if (L) is dropped, one can also find a solution that satisfies (E), (SYM), and (AC), but violates (L). Since its construction is involved, we refer the reader to Appendix~\ref{Appendix: AC} for details. Third, if (SYM) is dropped, consider the solution defined by $\varphi_1(v)=v(N)$ and $\varphi_i(v)=0$ for every $i\in N\setminus\{1\}$ and every $v\in\mathcal{V}$. This solution satisfies (E), (L), and (AC), but violates (SYM). Finally, the egalitarian value satisfies (E), (L), and (SYM), but violates (AC).

\subsection{Connections with Other Subclasses of the ELS Values}
We relate our theorems to existing results in the literature. First, \citet{wang2019family} show that a solution is an affine combination of $(\psi^s)_{s=1}^{n-1}$ if and only if it is an ELS value satisfying (IGP). 
Our Theorem~\ref{comp ELS}, by contrast, relies on (RNP), which is weaker than (IGP). Consequently, the family of solutions characterized by Theorem~\ref{comp ELS} includes the egalitarian value, whereas the family in \citet{wang2019family} does not. Moreover, our Theorem~\ref{consistency ELS} shows that, within the class of ELS values, (IGP) and (AC) are equivalent. As a counterpart to Lemma \ref{IGP_cons}, we now show the logical relationship between (AC) and (RNP). 

\begin{lemma}\label{AC_IGP}
(E), (L), and (AC) together imply (RNP). 
\end{lemma}

Second, affine combinations of $\{\psi^s\}_{s=1}^{n-1}$ include an important class of solutions, the \emph{least-square values} of \citet{ruiz1998family}. For $v\in\mathcal{V}$, $x\in\mathbb{R}^n$, and $S\subseteq N$, define the excess $e(S,x):=v(S)-\widehat{x}(S)$. Given a weight function $m:\{0,1,\ldots,n\}\to\mathbb{R}_{+}$, the \emph{least-square value with weights $m$} is obtained by
\[
\min_{x \in \mathbb{R}^n} \left\{\sum_{S \subseteq N}m\left(|S|\right) \left(e(S,x) - \overline{e(v,x)}\right)^2\right\} \ \text{s.t.} \ \sum_{i \in N} x_i = v(N), 
\]
where $\overline{e(v,x)} := \frac{1}{2^n-1}\sum_{\emptyset \neq S\subseteq N}e(S,x)$ is the average excess at $x$. 

\citet{ruiz1998family} establish a characterization of the least-square values based on (IGP) and the following monotonicity axiom.

\noindent\textbf{Coalitional Monotonicity (CM)}: 
For any $v,w\in\mathcal{V}$ and $T\subseteq N$, if $v(T)>w(T)$ and $v(S)=w(S)$ for all $S\neq T$, then $\varphi_i(v)\ge \varphi_i(w)$ for all $i\in T$.

\begin{theorem*}[\citealp{ruiz1998family}, Theorem 8]
Suppose that $\varphi$ is an ELS value. Then the following are equivalent:

\vspace{-0.3cm}

\begin{itemize}\itemsep-0.1cm
\item[(I)] $\varphi$ satisfies (IGP) and (CM).
\item[(II)] $\varphi$ is a least-square value.
\end{itemize}
\end{theorem*}

Because, within the ELS values, (AC) and (IGP) are equivalent, we obtain:

\begin{corollary}\label{cor:leastsquare}
Suppose that $\varphi$ is an ELS value. Then the following are equivalent:
\begin{itemize}\itemsep-0.1cm
\item[(I)] $\varphi$ satisfies (AC) and (CM).
\item[(II)] $\varphi$ is a least-square value.
\end{itemize}
\end{corollary}

Relatedly, \citet{wang2019family} characterize a subclass of the ELS values that satisfy (CM), which they refer to as the \emph{ideal values}.\footnote{See \citet{wang2019family} for the definition of ideal values.}

\begin{theorem*}[\citealp{wang2019family}, Theorem 3.2]
Suppose that $\varphi$ is an ELS value. Then the following are equivalent:

\vspace{-0.3cm}

\begin{itemize}\itemsep-0.1cm
\item[(I)] $\varphi$ satisfies (CM).
\item[(II)] $\varphi$ is an ideal value.
\end{itemize}
\end{theorem*}

A further related subclass of the ELS values is the class of \emph{procedural values} introduced by \citet{malawski2013ijgt}. In Malawski's framework, a value is procedural if and only if it satisfies (E), (L), equal treatment of equals, (CM), and weak monotonicity. Since equal treatment of equals is equivalent to (SYM) under (E) and (L), it follows that, within the class of ELS values, the procedural values are those ELS values that satisfy (CM) and weak monotonicity. Hence, the class of procedural values is a subclass of the ideal values of \citet{wang2019family}. Its intersection with the class of affine combinations of $\{\psi^s\}_{s=1}^{n-1}$ therefore consists of those affine combinations that also satisfy weak monotonicity. For example, this intersection includes the Shapley value.

In sum, the affine combinations of $\{\psi^s\}_{s=1}^{n-1}$ and the ideal values are both subclasses of the ELS values, and both contain the class of least-square values (Figure~\ref{fig:venn}). In the notation of this paper, the family $\{\psi^s\}_{s=1}^{n}$ spans the class of ELS values, while affine combinations of $\{\psi^s\}_{s=1}^{n-1}$ form a distinguished subclass. Theorems~\ref{comp ELS} and~\ref{consistency ELS} show that, within the class of ELS values, the composition axioms and active-player consistency identify this same affine-combination class, with the egalitarian value as the only additional possibility under the composition axioms. Under (CM) and (AC), this subclass further reduces to the least-square values.

\begin{figure}[H]
\centering
\tikzset{every picture/.style={line width=0.75pt}} 

\begin{tikzpicture}[x=0.75pt,y=0.75pt,yscale=-1,xscale=1]

\draw    (30,50) -- (220,50) ;
\draw    (310,50) -- (500,50) ;
\draw    (30,50) -- (30,350) ;
\draw    (500,50) -- (500,350) ;
\draw    (30,350) -- (500,350) ;

\begin{scope}
  \clip (133.36,113.92) .. controls (146.5,110.72) and (160.49,109) .. (175,109) .. controls (251.22,109) and (313,156.57) .. (313,215.25) .. controls (313,273.93) and (251.22,321.5) .. (175,321.5) .. controls (98.78,321.5) and (37,273.93) .. (37,215.25) .. controls (37,189.28) and (49.11,165.48) .. (69.21,147.02);

  \fill[gray!12] (459.17,151.45) .. controls (479.65,169.99) and (492,194.01) .. (492,220.25) .. controls (492,278.93) and (430.22,326.5) .. (354,326.5) .. controls (277.78,326.5) and (216,278.93) .. (216,220.25) .. controls (216,161.57) and (277.78,114) .. (354,114) .. controls (375.53,114) and (395.91,117.8) .. (414.07,124.57) -- cycle ;
\end{scope}

\draw  [draw opacity=0] (133.36,113.92) .. controls (146.5,110.72) and (160.49,109) .. (175,109) .. controls (251.22,109) and (313,156.57) .. (313,215.25) .. controls (313,273.93) and (251.22,321.5) .. (175,321.5) .. controls (98.78,321.5) and (37,273.93) .. (37,215.25) .. controls (37,189.28) and (49.11,165.48) .. (69.21,147.02) -- (175,215.25) -- cycle ; 
\draw   (133.36,113.92) .. controls (146.5,110.72) and (160.49,109) .. (175,109) .. controls (251.22,109) and (313,156.57) .. (313,215.25) .. controls (313,273.93) and (251.22,321.5) .. (175,321.5) .. controls (98.78,321.5) and (37,273.93) .. (37,215.25) .. controls (37,189.28) and (49.11,165.48) .. (69.21,147.02) ;  

\draw  [draw opacity=0] (459.17,151.45) .. controls (479.65,169.99) and (492,194.01) .. (492,220.25) .. controls (492,278.93) and (430.22,326.5) .. (354,326.5) .. controls (277.78,326.5) and (216,278.93) .. (216,220.25) .. controls (216,161.57) and (277.78,114) .. (354,114) .. controls (375.53,114) and (395.91,117.8) .. (414.07,124.57) -- (354,220.25) -- cycle ; 
\draw  [color={rgb, 255:red, 0; green, 0; blue, 0 }  ,draw opacity=1 ] (459.17,151.45) .. controls (479.65,169.99) and (492,194.01) .. (492,220.25) .. controls (492,278.93) and (430.22,326.5) .. (354,326.5) .. controls (277.78,326.5) and (216,278.93) .. (216,220.25) .. controls (216,161.57) and (277.78,114) .. (354,114) .. controls (375.53,114) and (395.91,117.8) .. (414.07,124.57) ;  

\draw  [draw opacity=0] (418.43,182.47) .. controls (434.9,195.01) and (445,211.91) .. (445,230.5) .. controls (445,269.16) and (401.35,300.5) .. (347.5,300.5) .. controls (293.65,300.5) and (250,269.16) .. (250,230.5) .. controls (250,191.84) and (293.65,160.5) .. (347.5,160.5) .. controls (356.02,160.5) and (364.28,161.28) .. (372.16,162.76) -- (347.5,230.5) -- cycle ; 
\draw   (418.43,182.47) .. controls (434.9,195.01) and (445,211.91) .. (445,230.5) .. controls (445,269.16) and (401.35,300.5) .. (347.5,300.5) .. controls (293.65,300.5) and (250,269.16) .. (250,230.5) .. controls (250,191.84) and (293.65,160.5) .. (347.5,160.5) .. controls (356.02,160.5) and (364.28,161.28) .. (372.16,162.76) ;  

\draw  [fill={rgb, 255:red, 0; green, 0; blue, 0 }  ,fill opacity=1 ] (232,214.5) .. controls (232,213.12) and (233.12,212) .. (234.5,212) .. controls (235.88,212) and (237,213.12) .. (237,214.5) .. controls (237,215.88) and (235.88,217) .. (234.5,217) .. controls (233.12,217) and (232,215.88) .. (232,214.5) -- cycle ;

\draw  [fill={rgb, 255:red, 0; green, 0; blue, 0 }  ,fill opacity=1 ] (277,217.5) .. controls (277,216.12) and (278.12,215) .. (279.5,215) .. controls (280.88,215) and (282,216.12) .. (282,217.5) .. controls (282,218.88) and (280.88,220) .. (279.5,220) .. controls (278.12,220) and (277,218.88) .. (277,217.5) -- cycle ;

\draw  [fill={rgb, 255:red, 0; green, 0; blue, 0 }  ,fill opacity=1 ] (362,234.5) .. controls (362,233.12) and (363.12,232) .. (364.5,232) .. controls (365.88,232) and (367,233.12) .. (367,234.5) .. controls (367,235.88) and (365.88,237) .. (364.5,237) .. controls (363.12,237) and (362,235.88) .. (362,234.5) -- cycle ;

\draw  [fill={rgb, 255:red, 0; green, 0; blue, 0 }  ,fill opacity=1 ] (129,205.5) .. controls (129,204.12) and (130.12,203) .. (131.5,203) .. controls (132.88,203) and (134,204.12) .. (134,205.5) .. controls (134,206.88) and (132.88,208) .. (131.5,208) .. controls (130.12,208) and (129,206.88) .. (129,205.5) -- cycle ;

\draw    (270,113) -- (270,148) ;
\draw [shift={(270,150)}, rotate = 270] [color={rgb, 255:red, 0; green, 0; blue, 0 }  ][line width=0.75]    (6.56,-1.97) .. controls (4.17,-0.84) and (1.99,-0.18) .. (0,0) .. controls (1.99,0.18) and (4.17,0.84) .. (6.56,1.97)   ;

\draw (232,41) node [anchor=north west][inner sep=0.75pt]  [font=\footnotesize] [align=left] {ELS $\displaystyle ( \psi ^{s})_{s=1}^{n}$};
\draw (69,112) node [anchor=north west][inner sep=0.75pt]  [font=\footnotesize] [align=left] {affine comb\\ \ of $\displaystyle (\psi^{s}){_{s=1}^{n-1}}$};
\draw (420,131) node [anchor=north west][inner sep=0.75pt]  [font=\footnotesize] [align=left] {ideal};
\draw (374.16,165.76) node [anchor=north west][inner sep=0.75pt]  [font=\footnotesize] [align=left] {procedual};
\draw (355.5,240.5) node [anchor=north west][inner sep=0.75pt]  [font=\footnotesize] [align=left] {ED};
\draw (222,221) node [anchor=north west][inner sep=0.75pt]  [font=\footnotesize] [align=left] {CIS};
\draw (259,223) node [anchor=north west][inner sep=0.75pt]  [font=\footnotesize] [align=left] {Shapley};
\draw (108,212.4) node [anchor=north west][inner sep=0.75pt]  [font=\footnotesize]  {$2\psi^{1} -\psi^{2}$};
\draw (238,95) node [anchor=north west][inner sep=0.75pt]  [font=\footnotesize] [align=left] {least-square};

\end{tikzpicture}
\caption{\small{Relationship between subclasses of the ELS values. 
As an example of a solution that is an affine combination of $\{\psi^s\}_{s=1}^{n-1}$ but not a least-square value, let $n>2$ and consider $\varphi:=2\psi^1-\psi^2$. 
To see that $\varphi$ is not a least-square value, let $n=3$ and let $v\in\mathcal{V}$ be the unanimity game $v=u_N$. Then $\psi^1(v)=\psi^2(v)=\varphi(v)=\left(\frac13,\frac13,\frac13\right)$. Next, define a game $w \in \mathcal{V}$ by $w(\{1,2\})=1$ and $w(S)=v(S)$ for all $S\neq \{1,2\}$. Then $\psi^1(w)=\left(\frac13,\frac13,\frac13\right)$ is unchanged, whereas $\psi^2(w)=ENSC(w)=\left(\frac23,\frac23,-\frac13\right)$. Hence, $\varphi(w)=2\psi^1(w)-\psi^2(w)=(0,0,1)$. Therefore, when the worth of coalition $\{1,2\}$ increases from $0$ to $1$, the payoffs of players $1$ and $2$ fall from $\frac13$ to $0$, so $\varphi$ violates (CM). Since every least-square value satisfies (CM), it follows that $\varphi$ is not a least-square value.}}
\label{fig:venn}
\end{figure}

\section{Nullified-Game Consistency}\label{sec_nullified}

We introduce three \emph{nullified-game consistency} axioms and characterize the Shapley, CIS, and ENSC values, each via one of these axioms. In this section, we assume $n \geq 3$ to exclude trivial cases. The idea of nullified-game consistency is stated as follows. When the recommended payoffs of a subset of players are fixed and those players are hypothetically treated as null, the solution restricted to the remaining players must reproduce their original payoffs. The population of players is unchanged, however, the nullified players whose shares have been predetermined no longer contribute to coalition worth. Accordingly, when the solution is applied to the reduced game, in which the contributions of these players are nullified, it must yield exactly the original payoffs for the remaining players. Formally, for $v \in \mathcal{V}$ and $S \subseteq N$, the \textbf{$\bm{S}$-nullified game} \citep{beal2014solidarity, beal2016nullification},\footnote{See \cite{abe2023potentials} for further discussions.} denoted $v|_S \in \mathcal{V}$, is defined by 
\[
v|_S(T) = v(T \cap S) \ \text{for all} \ T \subseteq N. 
\]
Specifically, note that for $i \in N$, 
\[
v|_{\{i\}}=v(\{i\})u_{\{i\}} \ \text{and} \ v^*|_{\{i\}}=v^*(\{i\})u_{\{i\}}=(v(N)-v(N\setminus \{i\}))u_{\{i\}}.
\]
For $S\subseteq N$, $v\in\mathcal{V}$, and $x\in\mathbb{R}^N$, let $R^{S}(x,v)$ denote the reduced game on the player set $S$, obtained by nullifying the players in $N\setminus S$ at the initial allocation $x$ in game $v$.
Note that $x$ is determined by the underlying rule $\varphi$.
Analogous to variable-population consistency, reduced games admit several definitions. We consider three such definitions and, accordingly, introduce three nullified-game consistency axioms, HM-, F-, or M-nullified game consistency, each corresponding to each reduced game definition.

\noindent\textbf{HM-Nullified Game Consistency (HM-NGC)}:  For any $v \in \mathcal{V}$, $S \subseteq N$ with $|S| \ge 2$, and $i \in S$, 
\[
\varphi_i(v)=\varphi_i\left(R^{\text{HM},S}\left(\varphi(v), v\right)\right), \ \text{where}
\]
\begin{eqnarray*}
R^{\text{HM},S}(\varphi(v), v)(T) &=&
\begin{cases}
v(T\cup (N\setminus S))-\sum_{j \in N\setminus S}\varphi_j(v|_{T \cup (N\setminus S)}) &\text{if} \ T \cap S \neq \emptyset,\\
0 &\text{otherwise}.
\end{cases} 
\end{eqnarray*}
F-nullified game consistency (F-NGC) and M-nullified game consistency (M-NGC) are defined analogously, using the following reduced games:
\begin{eqnarray*}
R^{\text{F},S}(\varphi(v), v)(T) &=&
\begin{cases}
v(N)-\sum_{j \in N\setminus S}\varphi_j(v|_{\{j\}}) &\text{if} \ S \subseteq T,\\
v(T\cap S) &\text{otherwise},
\end{cases} \\
R^{\text{M},S}(\varphi(v), v)(T)&=&
\begin{cases}
v(T \cup (N\setminus S))-\sum_{j \in N\setminus S}\varphi_j(v^*|_{\{j\}}) &\text{if} \ T \cap S \neq \emptyset,\\
0 &\text{otherwise}.
\end{cases}
\end{eqnarray*}




\medskip

The idea of nullified-game consistency was introduced by \citet{kn2025el} for set-valued solution concepts. They proposed three consistency axioms and, in each case, characterized the core. In our setting, (F-NGC) and (M-NGC) are the single-valued counterparts of their \emph{Funaki nullified reduced game property} and \emph{Moulin nullified reduced game property}, respectively. The reduced game used in (F-NGC) is inspired by the projection-reduced game of \citet{funaki1996wp}, whereas that in (M-NGC) follows the complement-reduced game of \citet{moulin1985separability}.

In each case, players in $N \setminus S$ are null players in both the $S$-nullified game and the corresponding reduced games.
The distinction between these concepts lies in the role assigned to such nullified players.
In an $S$-nullified game, players in $N\setminus S$ neither generate value nor participate in any cooperation with players in \(S\). 
In contrast, in reduced games, although players in $N \setminus S$ do not generate value themselves, they may still cooperate with players in $S$ in exchange for a specified payoff.
As noted by \cite{kn2025el}, the differences can be interpreted in terms of the type of ``contract'' formed among players. 
Such contracts determine both the payoff allocated to the nullified players and the manner in which they cooperate with players in $S$.

A key distinction from the properties for set-valued solutions is how the payoffs of players in $N\setminus S$ are fixed: \citet{kn2025el} anchor them by an arbitrary selection from the solution set, whereas in our single-valued setting they are determined various ways.
More precisely, it is canonically by the single-valued solution $\varphi$, via $\varphi_j\left(v|_{\{j\}}\right)$ in (F-NGC) and $\varphi_j\left(v^*|_{\{j\}}\right)$ in (M-NGC). Finally, (HM-NGC) is a new notion inspired by the variable-population consistency property of \citet{hart1989potential}: for some $T\subseteq S$, each player $j\in N\setminus S$ is assigned $\varphi_j\left(v|_{T\cup (N\setminus S)}\right)$, i.e., the players in $N\setminus S$ act as a single block that cooperates with a subset of $S$. 
It is straightforward to verify that, if  the corresponding term in each axiom is replaced by $\varphi_i(v)$, the class of solutions satisfying these axioms is entirely different from the one characterized in the main result of this section.

The final axiom, \emph{equal gain for two players}, states that when everyone except two players $i$ and $j$ is a null player, the solution should give $i$ and $j$ the same extra amount over their stand-alone worths. In other words, if only $i$ and $j$ matter in the game, the solution treats them symmetrically by awarding them equal incremental gains above what they could secure alone.

\medskip

\noindent \textbf{Equal gain for two players (EG)}: For any $v \in \mathcal{V}$ and $i,j \in N$, if any $k \neq i,j$ is a null player in $v$, then $\varphi_i(v)-v(\{i\})=\varphi_j(v)-v(\{j\})$.

\medskip



We now present characterizations of the Shapley, CIS, and ENSC values, each obtained from one of the consistency axioms defined above. 
\begin{theorem}\label{NGC axiomatization}
A solution satisfies: 

\vspace{-0.3cm} 

\begin{itemize}\itemsep-0.1cm
\item[(I)] (E), (EG), and (HM-NGC) if and only if it is the Shapley value;
\item[(II)] (E), (EG), and (F-NGC) if and only if it is the CIS value; and 
\item[(III)] (E), (EG), and (M-NGC) if and only if it is the ENSC value. 
\end{itemize}
\end{theorem}

In variable-population frameworks, \emph{two-player standardness} is an axiom often used to obtain a characterization of a solution via consistency.\footnote{For example, \cite{hart1989potential}, \cite{driessenfunaki1997}, and \cite{vcfp2016td}.} Since our consistency is formulated for a fixed population, (EG) serves as the corresponding fixed-population analogue of two-player standardness.

We conclude this section by providing examples to demonstrate the independence of the axioms used in Theorem \ref{NGC axiomatization}. First we drop (E). Define $\varphi$ as follows: for any $v \in \mathcal{V}$ and $i \in N$, $\varphi_i(v) = v(\{i\})$. Then for any $v \in \mathcal{V}$ and $i,j \in N$ such that any $k \neq i,j$ is a null player in $v$, $\varphi_i(v) - v(\{i\}) = 0 = \varphi_j(v) - v(\{j\})$. Thus $\varphi$ satisfies (EG). For any $v \in \mathcal{V}$, $S \subseteq N$ with $|S| \geq 2$, and $i \in S$, $R^{\text{F},S}(\varphi(v),v)(\{i\}) = v(\{i\})$. Thus for any $v \in \mathcal{V}$, $S \subseteq N$ with $|S| \geq 2$, and $i \in S$, $\varphi_i\left(R^{\text{F},S}(\varphi(v),v)\right) = v(\{i\}) = \varphi_i(v)$, implying that $\varphi$ satisfies (F-NGC). On the other hand, it is trivial to verify that $\varphi$ violates (E). 

Now define $\varphi$ as follows: for any $v \in \mathcal{V}$ and $i \in N$, $\varphi_i(v) = v(N) - v(N \setminus \{i\})$. Then for any $v \in \mathcal{V}$ and $i,j \in N$ such that any $k \neq i,j$ is a null player in $v$, $\varphi_i(v) - v(\{i\}) = \left(v(\{i,j\}) - v(\{j\})\right) - v(\{i\}) = \left(v(\{i,j\}) - v(\{i\})\right) - v(\{j\}) = \varphi_j(v) - v(\{j\})$. Thus $\varphi$ satisfies (EG). For any $v \in \mathcal{V}$ and $j \in N$, 
\[
\varphi_j\left(v^*|_{\{j\}}\right) = v^*|_{\{j\}}(N) - v^*|_{\{j\}}(N \setminus \{j\}).
\]
Because
\[
v^*|_{\{j\}}(N \setminus \{j\}) = \left(v(N) - v(N \setminus \{j\})\right)u_{\{j\}}(N \setminus \{j\}) = 0,
\]
we have 
\[
\varphi_j\left(v^*|_{\{j\}}\right) = v(N) - v(N \setminus \{j\}). 
\]
Then, by definition of $\varphi$, for any $\text{I} \in \{\text{HM},\text{M}\}$, $v \in \mathcal{V}$, $S \subseteq N$ with $|S| \geq 2$, and $i \in S$, 
\begin{eqnarray*}
\varphi_i\left(R^{\text{I}, S}\left(\varphi(v),v\right)\right) &=& R^{\text{I},S}\left(\varphi(v),v\right)(N) - R^{\text{I},S}\left(\varphi(v),v\right)\left(N \setminus \{i\}\right) \\
&=& v(N) - \sum_{j \in N \setminus S} \varphi_j\left(v^*|_{\{j\}}\right) - \left(v(N \setminus \{i\}) - \sum_{j \in N \setminus S} \varphi_j\left(v^*|_{\{j\}}\right)\right) \\
&=& v(N) - v(N \setminus \{i\}) = \varphi_i(v). 
\end{eqnarray*}
Therefore, it satisfies (HM-NGC) and (M-NGC). On the other hand, it is again trivial to verify that $\varphi$ violates (E). 

Second we drop (EG). Define $\varphi$ as follows: there is $i \in N$ such that for any $v \in \mathcal{V}$, $\varphi_i(v) = v(N)$, and for any $j \neq i$, $\varphi_j(v) = 0$. Then it is trivial to verify that while $\varphi$ satisfies (E), it violates (EG). To see that $\varphi$ satisfies (I-NGC), where $\text{I} \in \{\text{HM},\text{F}, \text{M}\}$, consider $v \in \mathcal{V}$ and $S \subseteq N$ with $|S| \geq 2$. Then if $i \in S$, then $R^{\text{I},S}(\varphi(v),v)(N) = v(N)$. Thus $\varphi_i\left(R^{\text{I},S}(\varphi(v),v)\right) = v(N) = \varphi_i(v)$, and for any $j \in S \setminus \{i\}$, $\varphi_j\left(R^{\text{I},S}(\varphi(v),v)\right) = 0 = \varphi_j(v)$. If $i \notin S$, then for any $j \in S$, $\varphi_j\left(R^{\text{I},S}(\varphi(v),v)\right) = 0 = \varphi_j(v)$. Therefore, $\varphi$ satisfies (I-NGC). 

Finally, if any one of (HM-NGC), (F-NGC), or (M-NGC) is dropped, the remaining axioms (E) and (EG) are still satisfied by the other two values. Specifically:

\vspace{-0.3cm}

\begin{itemize}\itemsep-0.1cm
\item without (HM-NGC), the CIS and ENSC values satisfy (E) and (EG);
\item without (F-NGC), the Shapley and ENSC values satisfy (E) and (EG); and
\item without (M-NGC), the Shapley and CIS values satisfy (E) and (EG).
\end{itemize}

\section{Conclusion}\label{sec_conclude}
This paper studies solutions for TU-games from a fixed-population perspective. We keep the set of players fixed and ask how a solution should behave when parts of the outcome are treated as already settled. The guiding idea is that certain components of the allocation can be fixed in advance, for example, by provisional divisions of worth, by guaranteed terms for some players, or by commitments that render some players effectively neutral, and the solution is then required to handle the remaining part of the problem in a disciplined way.

Within this framework, we obtain structural results for the ELS values. We show that every solution in this class can be represented as a Shapley value applied to a transformed game, and we use several invariance requirements to locate familiar solutions within that structure, including the Shapley, CIS, ENSC, and egalitarian values, as well as least-square values. Rather than treating these solutions as isolated objects, the analysis explains how they arise from different ways of holding parts of the allocation fixed and resolving what remains.

This perspective has two consequences. First, it offers a rationale for benchmark solutions in environments where the relevant group of players is effectively given (for instance, divisions within firms, committees, or jurisdictions), so that the question is not who participates, but how the joint surplus is divided. Second, it suggests a way to assess alternative solutions: not only by standard axioms such as efficiency or symmetry, but also by how they respond when some payoffs are locked in and the rest of the problem must still be solved.

A natural direction for future research is to move beyond the ELS values and ask to what extent these fixed-population requirements can serve as a systematic guideline for designing desirable solutions more generally.

\begin{center}
\Large{{\bf Appendix: Omitted Proofs}}
\end{center}

\begin{appendix}

\numberwithin{theorem}{section}
\numberwithin{lemma}{section}
\numberwithin{proposition}{section}
\numberwithin{example}{section}
\renewcommand{\theequation}{\thesection.\arabic{equation}}
\appendix




\section{Proofs of Results in Section \ref{sec_sh}}\label{appendix: sec_sh}
\begin{proof}[Proof of Lemma \ref{SYML sigma-sh}]
Fix $\sigma: \{1,\ldots,n\} \to \mathbb{R}$ and consider the $\sigma$-Shapley value. By definition, it satisfies (L) and (SYM).
Moreover, for any $v \in \mathcal{V}$ and $i \in N$,
\begin{eqnarray*}
\sigma\text{-}Sh_i(v) = Sh_i(v_\sigma)&=&\sum_{S \subseteq N: i \notin S}\frac{|S|!(n-|S|-1)!}{n!} \left( \sigma(|S|+1)v(S\cup \{i\})-\sigma(|S|)v(S) \right)\\
&=&\sum_{S \subseteq N: i \in S}\frac{(|S|-1)!(n-|S|)!}{n!}\sigma(|S|)v(S)-\sum_{S \subseteq N: i \notin S}\frac{|S|!(n-|S|-1)!}{n!}\sigma(|S|)v(S).
\end{eqnarray*}
It follows that the $\sigma$-Shapley value admits the representation in Lemma \ref{Linear} (ii), with coefficients satisfying $q_k = -\frac{k}{n-k}p_k$ for all $k=1,\ldots, n-1$. 

Conversely, suppose that $\varphi$ satisfies (L) and (SYM), and that in the representation given in Lemma~\ref{Linear} (ii), the coefficients satisfy $q_k = -\frac{k}{\,n-k\,} p_k$ for all $k = 1,\ldots,n-1$. By Lemma \ref{Linear} (ii),
\begin{eqnarray*}
\varphi_i(v)&=&\sum_{S \subseteq N: i \in S}   p_sv(S)+\sum_{S \subseteq N: i \notin S}  q_sv(S)\\
&=&\sum_{S \subseteq N: i \notin S}\frac{|S|!(n-|S|-1)!}{n!}np_s\left(\frac{(n-1)!}{|S|!(n-|S|-1)!}v(S \cup \{i\})-\frac{(n-1)!}{(|S|-1)!(n-|S|)!}v(S) \right).
\end{eqnarray*}
Let $\sigma(s):=n\binom{n-1}{s-1} p_s$ for $s=1,\ldots,n-1$. 
Then the representation becomes 
\[
\varphi_i(v)=\sum_{S \subseteq N: i \notin S}\frac{|S|!(n-|S|-1)!}{n!}\left(\sigma(|S|+1)v(S \cup \{i\})-\sigma(|S|)v(S) \right)=Sh_i(v_\sigma).
\]
\end{proof}

\section{Proofs of Results in Section \ref{sec_consistency}}\label{appendix:sec_consistency}

\begin{proof}[Proof of Lemma \ref{IGP_s}]
Fix $s \in \{1,\ldots,n-1\}.$ Recall that  $\psi ^{s}$ is defined as follows: for any $v \in \mathcal{V}$ and $i \in N$, 
\[
\psi _{i}^{s}(v)=\frac{v\left(N\right) }{n}+\frac{n-1}{s}\left( \frac{%
\sum_{S:\left|S\right| =s}v\left( S\right) }{\binom{n}{s}}-\frac{%
\sum_{S:\left\vert S\right\vert =s, i \notin S}v\left( S\right) }{\binom{n-1}{%
s}}\right). 
\]
By denoting $V^+_i = \sum_{S:\left\vert S\right\vert =s: i \in S}v\left( S\right) $ and $V^-_i = \sum_{S:\left\vert S\right\vert =s: i \notin S}v\left( S\right)$, $\psi ^{s}$ can be written as:
\begin{eqnarray*}
\psi _{i}^{s}(v) &=&\frac{v\left( N\right) }{n}+\frac{n-1}{s}\frac{1}{\binom{%
n}{s}}\left( V_{i}^{+}+V_{i}^{-}-\frac{\binom{n}{s}}{\binom{n-1}{s}}%
V_{i}^{-}\right)  \\
&=&\frac{v\left( N\right) }{n}+\frac{n-1}{s}\frac{1}{\binom{n}{s}}\left(
V_{i}^{+}-\frac{s}{n-s}V_{i}^{-}\right)  \\
&=&\frac{v\left( N\right) }{n}+\frac{n-1}{n}\frac{1}{\binom{n-1}{s-1}}\left(
V_{i}^{+}-\frac{s}{n-s}V_{i}^{-}\right) .
\end{eqnarray*}%
Let $x \in \mathbb{R}^n$ and $v=\widehat{x}$. Then, we have $\frac{v(N)}{n}=\overline{x}$ and 
\begin{eqnarray*}
V_{i}^{+}-\frac{s}{n-s}V_{i}^{-}&=& \binom{n-1}{s-1}x_i+\binom{n-2}{s-2}\sum_{j \neq i}x_j-\frac{s}{n-s} \binom{n-2}{s-1}\sum_{j \neq i}x_j\\
&=&  \binom{n-1}{s-1}\Bigl(x_i-\frac{n}{n-1}\sum_{j \neq i}x_j    \Bigr)\\
&=& \binom{n-1}{s-1}\frac{n}{n-1}(x_i-\overline{x}).
\end{eqnarray*}
This implies that $\varphi^s_i(\widehat{x})=\overline{x}+ (x_i-\overline{x})=x_i$.
\end{proof}

\medskip

\begin{proof}[Proof of Lemma \ref{IGP_cons}]
(i) By definition of $U$, $U(\varphi(v),v)=v-\widehat{\varphi(v)}$. Letting $t=v(N)$ in (CU), $\varphi(v)=\varphi(v)+\varphi(v-\widehat{\varphi(v)})$, and thus $\varphi(v-\widehat{\varphi(v)})=0$. By (L), $
\varphi(v)=\varphi(\widehat{\varphi(v)}).$

\medskip

\noindent (ii) and (iii) First, (E) implies that 
$D_O(\varphi(v),v)
=D_I(\varphi(v),v)
=\widehat{\varphi(v)}.$
Hence, by letting $t=v(N)$ in (CD$_{\text{O}}$) and (CD$_{\text{I}}$), we have
$\varphi(v)
=\varphi(D_O(\varphi(v)))
=\varphi(D_I(\varphi(v)))
=\varphi(\widehat{\varphi(v)}).$
\end{proof}

\medskip

\begin{proof}[Proof of Lemma \ref{AC_IGP}]
Suppose that $\varphi$ satisfies (E), (L), and (AC). By (L) and (AC), for any $v \in \mathcal{V}$, $S \subsetneq N$, and $i \in S$, $\varphi_i\left(\widehat{\varphi^S(v)}\right) = 0$, where $\varphi^S(v) = \left((0)_{i \in S}, (\varphi_i(v))_{i \notin S}\right) \in \mathbb{R}^n$. This together with (L) implies that for any $v \in \mathcal{V}$ and $i \in N$,
\[
\varphi_i\left(\widehat{\varphi(v)}\right) = \sum_{k \in N} \varphi_i\left(\widehat{\varphi^{N \setminus \{k\}}(v)}\right) = \varphi_i\left(\widehat{\varphi^{N \setminus \{i\}}(v)}\right). 
\]
By (E) and the fact that $\varphi_k\left(\widehat{\varphi^{N \setminus \{i\}}(v)}\right) = 0$ for all $k \neq i$, we have 
\[
\varphi_i\left(\widehat{\varphi^{N \setminus \{i\}}(v)}\right) = \varphi_i(v). 
\]
Thus $\varphi_i\left(\widehat{\varphi(v)}\right) = \varphi_i(v)$, implying that $\varphi$ satisfies (RNP). 

\end{proof}

\section{Proof of Theorem \ref{NGC axiomatization}}
For necessity, it suffices to show that the Shapley, CIS, and ENSC values satisfy their respective nullified-game consistency axioms.
\begin{lemma}
The following statements hold. 

\vspace{-0.3cm}

\begin{itemize}\itemsep-0.1cm
\item[(i)] The Shapley value satisfies (HM-NGC).
\item[(ii)] The CIS value satisfies (F-NGC).
\item[(iii)] The ENSC value satisfies (M-NGC). 
\end{itemize}
\end{lemma}
\begin{proof}
(i) We show that, for any $v \in \mathcal{V}$ and $S \subseteq N$ with $|S|\ge2$, 
\begin{equation}
P\left(R^{\text{HM},S}(\varphi(v), v)|_{S'}\right)=P\left(v|_{S' \cup (N\setminus S)}\right)-P\left(v|_{N\setminus S}\right)~\forall S' \subseteq N,  \label{potential eq}
\end{equation}
where $P: \mathcal{V} \to \mathbb{R}$ is the potential function of the Shapley value 
$$
P(v)=\sum_{S\subseteq N}\frac{(|S|-1)!(n-|S|)!}{n!}v(S).
$$
\cite{abe2023potentials} show that the Shapley value is uniquely represented as $Sh_i(v)=P(v)-P\left(v|_{N\setminus \{i\}}\right)$.\footnote{The potential function of TU-games is originally defined by \cite{hart1989potential}. For recent studies, see \cite{ortmann1998conservation}, \cite{casajus2014potential}, \cite{casajus2015potential, casajus2018decomposition}, and references theirn.}
Then by this formula and (\ref{potential eq}), we can see that, for any $i \in S$,
\begin{eqnarray*}
\varphi_i\left(R^{\text{HM},S}(\varphi(v),v)\right)&=&P\left(R^{\text{HM},S}(\varphi(v), v)\right)-P\left(R^{\text{HM},S}(\varphi(v), v)|_{N\setminus \{i\}}\right)\\
&=& \left(P(v)-P(v|_{N\setminus S})\right)-\left(P(v|_{N\setminus \{i\}})-P(v|_{N\setminus S})  \right)\\
&=&P(v)-P\left(v|_{N\setminus \{i\}}\right)\\
&=&\varphi_i(v).
\end{eqnarray*}

Now, we show (\ref{potential eq}) by induction.
Take any $v \in \mathcal{V}$ and suppose that $|S'|=1$.
Let us denote $S'=\{i_1\}$ for some $i_1 \in N$.
Then, by definition, $R^{\text{HM},S}\left(\varphi(v), v\right)|_{\{i_1\}}(N)=\varphi_{i_1}\left(v|_{\{i_1\} \cup (N\setminus S)}\right)$.
Since any player $j \neq i_1$ is null in $R^{\text{HM},S}(\varphi(v), v)|_{\{i_1\}}$, we have
\begin{align*}
P\left(R^{\text{HM},S}(\varphi(v), v)|_{\{i_1\}}\right)-P\left(R^{\text{HM},S}(\varphi(v), v)|_{\emptyset}\right)&=\varphi_{i_1}(R^{\text{HM},S}\left(\varphi(v), v)|_{\{i_1\}}\right)\\
&=\varphi_{i_1}\left(v|_{\{i_1\} \cup (N\setminus S)}\right)\\
&=P(v|_{\{i_1\} \cup (N\setminus S)})-P(v|_{N\setminus S}). 
\end{align*}
Then, because $P\left(R^{\text{HM},S}(\varphi(v), v)|_{\emptyset}\right)=0$, we have  $P\left(R^{\text{HM},S}(\varphi(v), v)|_{\{i_1\}}\right)=P\left(v|_{\{i_1\} \cup (N\setminus S)})-P(v|_{N\setminus S}\right)$.

Suppose that \eqref{potential eq} holds for any $S' \subseteq N$ with $|S'| \le k \in \{1, \ldots, n-1\}$. Let $S'$ with $|S'|=k+1$.
By definition, $R^{\text{HM},S}\left(\varphi(v), v\right)|_{S'}(N)=\sum_{j \in S'}\varphi_j\left(v|_{S' \cup (N\setminus S)}\right)$.
Moreover, since any player $j \notin S'$ is null in $R^{\text{HM},S}(\varphi(v), v)|_{S'}$, we have $\sum_{j \in S'}\varphi_j\left(R^{\text{HM},S}(\varphi(v), v)|_{S'}\right)=R^{\text{HM},S}(\varphi(v), v)|_{S'}(N)$.
By the induction hypothesis, rearranging these two equations leads to 
\begin{eqnarray*}
&&|S'|P\left(R^{\text{HM},S}(\varphi(v), v)|_{S'}\right)-\sum_{j \in S'}P\left(R^{\text{HM},S}(\varphi(v), v)|_{S'\setminus \{j\}}\right)=|S'|P\left(v|_{S' \cup (N\setminus S)}\right)-\sum_{j \in S'}P(v|_{(S' \setminus \{j\}) \cup (N\setminus S)}))\\
&\Leftrightarrow& |S'|P\left(R^{\text{HM},S}(\varphi(v), v)|_{S'}\right)=|S'|P\left(v|_{S' \cup (N\setminus S)}\right)+\sum_{j \in S'}\left(P\left(R^{\text{HM},S}(\varphi(v), v)|_{S'\setminus \{j\}}\right)- P\left(v|_{(S' \setminus \{j\}) \cup (N\setminus S)}\right)\right)\\
&\Leftrightarrow&|S'|P\left(R^{\text{HM},S}(\varphi(v), v)|_{S'}\right)=|S'|P(v|_{S' \cup (N\setminus S)})-|S'|P(v|_{N\setminus S}). 
\end{eqnarray*}
Thus \eqref{potential eq} holds for $|S'|=k+1$.

\medskip

\noindent (ii) Take any $v \in \mathcal{V}$ and $S \subseteq N$ with $|S| \geq 2$.
By definition, we have
\[
R^{\text{F},S}(CIS(v), v)(\{i\})=
\begin{cases}
v(\{i\}) &\text{if} \ i \in S,\\
0 &\text{otherwise}.
\end{cases}
\]
Moreover, 
$$
R^{\text{F},S}(CIS(v), v)(N)=v(N)-\sum_{j \in N\setminus S}CIS_j(v|_{\{j\}})=v(N)-\sum_{j \in N\setminus S}v(\{j\})CIS_j(u_{\{j\}})=v(N)-\sum_{j \in N \setminus S}v(\{j\}).
$$
Therefore, for any $i \in S$,
\begin{eqnarray*}
CIS_i(R^{\text{F},S}(CIS(v), v))&=&R^{\text{F},S}(CIS(v), v)(\{i\})+\frac{1}{n}\left(R^{\text{F},S}(CIS(v), v)(N)-\sum_{j \in N}R^{\text{F},S}(CIS(v), v)(\{j\})\right)\\
&=&v(\{i\})+\frac{1}{n}\left(v(N)-\sum_{j \in N \setminus S}v(\{j\})-\sum_{j \in S}v(\{j\})\right)\\
&=&CIS_i(v).
\end{eqnarray*}

\medskip

\noindent (iii) Take any $v \in \mathcal{V}$ and $S \subseteq N$ with $|S| \geq 2$.
By definition, for any $j \in N \setminus S$, 
$$
ENSC_j\left(v^*|_{\{j\}}\right)=\left(v(N)-v(N\setminus \{j\})\right)ENSC_j(u_{\{j\}})=v(N)-v(N\setminus \{j\}).
$$
Hence, $R^{\text{M},S}(ENSC(v), v)(N)=v(N)-\sum_{j \in N\setminus S}\left(v(N)-v(N\setminus \{j\})  \right)$ and 
\[
R^{\text{M},S}(\varphi(v), v)(N\setminus \{j\})=
\begin{cases}
v(N\setminus \{j\})-\sum_{j \in N\setminus S}\left( v(N)-v(N\setminus \{j\})\right) &\text{if} \ j \in S\\
R^{\text{M},S}(\varphi(v), v)(N) &\text{otherwise}.
\end{cases}
\]
Therefore, for any $i \in S$,
\begin{align*}
&ENSC_i\left(R^{\text{M},S}(ENSC(v), v)\right)\\
&\hspace{2.5cm}=\left(R^{\text{M},S}(\varphi(v), v)(N)-R^{\text{M},S}(\varphi(v), v)\left(N\setminus \{i\}\right)\right)\\
&\hspace{3cm}+\frac{1}{n}\left(R^{\text{M},S}(\varphi(v), v)(N)-\sum_{j \in N}\left(R^{\text{M},S}(\varphi(v), v)(N)-R^{M,S}(\varphi(v), v)(N\setminus \{j\}) \right)\right)\\
&\hspace{2.5cm}=v(N)-v(N\setminus \{i\})+\frac{1}{n}\left(v(N)-\sum_{j \in N}\left(v(N)-v(N\setminus \{i\})\right) \right)\\
&\hspace{2.5cm}=ENSC_i(v).
\end{align*}
\end{proof}

The sufficiency of (I) requires two auxiliary lemmas, stated below.

\noindent \textbf{Minimal Right (MR)}:  For any $v \in \mathcal{V}$ and $i \in N$, if any $j \neq i$ is null in $v$, then $\varphi_i(v)=v(N)$.

\begin{lemma}\label{COV_EG_MG}
(E) and (EG) together imply (MR). 
\end{lemma}

\begin{proof}

Let $v\in \mathcal{V}$ and $i \in N$ be such that $N\setminus \text{Null}(v)=\{i\}$.
For $j \neq i$, by (EG), 
$$
\varphi_i(v)-v(\{i\})=\varphi_j(v)-v(\{j\})=\varphi_j(v).
$$
Then by (E), 
\begin{eqnarray*}
v(N) = \sum_{j \in N}\varphi_j(v)&=& \varphi_i(v) + \sum_{j \neq i}\varphi_j(v) \\
&=&  \varphi_i(v)+\sum_{j \neq i}\left(\varphi_i(v)-v(\{i\}\right)) \\
&=& \varphi_i(v) + (n-1)\left(\varphi_i(v) - v(\{i\})\right).  \end{eqnarray*}
Because any $j \neq i$ is null in  $v$, $v(N) = v(\{i\})$. This implies that 
$$
v(\{i\}) = \varphi_i(v) + (n-1)\left(\varphi_i(v) - v(\{i\})\right) \Leftrightarrow \varphi_i(v) = v(\{i\}). 
$$
Overall, $\varphi_i(v) = v(N)$. 
\end{proof}


\begin{lemma}\label{null player set HM}
Suppose that $\varphi$ satisfies (E), (MR), (HM-NGC). If player $i \in N$ is null in $v$, then player $i$ is also null in $R^{\text{HM}, N\setminus \{j\}}(\varphi(v), v)$ for any $j \neq i$.
\end{lemma}
\begin{proof}
Let $v \in \mathcal{V}$ and $i \in N$ be such that player $i$ is null in $v$.
For any $j \neq i$ and $T \subseteq N \setminus \{i\}$,
\[
R^{\text{HM},N \setminus \{j\}}(\varphi(v), v)(T \cup \{i\})-R^{\text{HM},N \setminus \{j\}}(\varphi(v), v)(T)=
\begin{cases}
v(\{i,j\})-\varphi_j(v|_{\{i,j\}}) &\text{if} \ T=\{j\},\\
\left( v(T \cup \{i,j\})-v(T \cup \{j\}) \right)\\~~~~~~~~~-\left(\varphi_j(v|_{T \cup \{i,j\}})
-\varphi_j(v|_{T \cup \{j\}}) \right) &\text{if} \ T \neq \{j\}.
\end{cases}
\]
If $T=\{j\}$, since player $i$ is null in $v$ and $\varphi$ satisfies (MR), we have 
$$
v(\{i,j\})-\varphi_j(v|_{\{i,j\}})=v(\{j\})-\varphi_j(v|_{\{j\}})=v(\{j\})-v(\{j\})=0.
$$
Similarly, if $T \neq \{j\}$, 
$$
\left( v(T \cup \{i,j\})-v(T \cup \{j\}) \right)-\left( \varphi_j(v|_{T \cup \{i, j\}})-\varphi_j(v|_{T \cup \{j\}}) \right)=0-\left( \varphi_j(v|_{T \cup \{j\}})-\varphi_j(v|_{T \cup \{j\}}) \right)=0.
$$
Thus, player $i$ is null in $R^{\text{HM}, N\setminus \{j\}}(\varphi(v), v)$.
\end{proof}

\medskip

\begin{lemma}
The following statements hold.

\vspace{-0.3cm}

\begin{itemize}\itemsep-0.1cm
\item[(i)] If $\varphi$ satisfies (E), (EG), and (HM-NGC), then $\varphi = Sh$. 
\item[(ii)] If $\varphi$ satisfies (E), (EG), and (F-NGC), then $\varphi = CIS$.
\item[(iii)] If $\varphi$ satisfies (E), (EG), and (M-NGC), then $\varphi = ENSC$.  
\end{itemize}
\end{lemma}
\begin{proof}
\noindent (i) Take any $v \in \mathcal{V}$ and $i,j \in N$.
By construction, any $k \neq i,j$ is a null player in $R^{\text{HM},\{i,j\}}\left(\varphi(v),v\right)$.
Hence, by (HM-NGC), (EG) and (E), we have  
\begin{eqnarray*}
\varphi_i(v)-\varphi_j(v)&=&\varphi_i\left(R^{\text{HM},\{i,j\}}(\varphi(v),v)\right)-\varphi_j\left(R^{\text{HM},\{i,j\}}(\varphi(v),v)\right)\\
&=&R^{\text{HM},\{i,j\}}(\varphi(v),v)(\{i\})-R^{\text{HM},\{i,j\}}(\varphi(v),v)(\{j\})\\
&=&\left(v(N\setminus \{j\})-\sum_{k \neq i,j}\varphi_k(v|_{N\setminus \{j\}})\right)-\left( v(N\setminus \{i\})-\sum_{k \neq i,j}\varphi_k(v|_{N\setminus \{i\}})\right)\\
&=&\left( \varphi_i(v|_{N\setminus \{j\}})+\varphi_j(v|_{N\setminus \{j\}})\right)-\left( \varphi_i(v|_{N\setminus \{i\}})+\varphi_j(v|_{N\setminus \{i\}}) \right).
\end{eqnarray*}
Note that $|\text{Null}(v|_{N\setminus \{k\}})| \ge |\text{Null}(v)|$ where $k=i,j$.
Therefore, we show the uniqueness of the solution by induction on the number of null players in $v$.
We divide the proof into three cases.

\noindent Case 1: $|\text{Null}(v)| \ge n-1$.
In this case, $\varphi$ is uniquely determined by (E), (EG), and (MR).\\
Case 2: $|\text{Null}(v)|=n-2$. For $\{i,j\}=N\setminus \text{Null}(v)$, by (EG), $\varphi_i(v)-\varphi_j(v)=v(\{i\})-v(\{j\})$.
Moreover, by (HM-NGC) and (MR), for any $k \in \text{Null}(v)$,
$$
\varphi_i(v)=\varphi_i\left(R^{\text{HM},\{i,k\}}(\varphi(v),v)\right)=R^{\text{HM},\{i,k\}}(\varphi(v),v)(N)=v(N)-\sum_{h \neq i,k}\varphi_h(v|_N)=v(\{i,j\})-\varphi_j(v).
$$ 
Therefore, $\varphi_i(v)$ and $\varphi_j(v)$ are uniquely determined.

\noindent\textbf{Induction hypothesis}: Suppose that $\varphi$ is uniquely determined for any $v\in \mathcal{V}$ with $|\text{Null}(v)| \ge k$, where $k \in \{1, \ldots, n\}$.
\\
Case 3: $|\text{Null}(v)|=k-1$.
Take any $i \in N\setminus \text{Null}(v)$ and consider $R^{\text{HM},N\setminus\{i\}}(\varphi(v),v)$.
By Lemma \ref{null player set HM}, $|\text{Null} (R^{\text{HM},N\setminus\{i\}}(\varphi(v),v))| \ge k$.
Since $\varphi$ satisfies (HM-NGC), $\varphi_j(v)=\varphi_j\left(R^{\text{HM},N\setminus\{i\}}(\varphi(v),v)\right)$ for any $j \neq i$. 
Moreover, by the induction hypothesis, the right-hand side is uniquely determined.
Therefore, by (E), $\varphi_i(v)$ is also uniquely determined, and hence, $\varphi$ is uniquely determined for $v$ with $|\text{Null}(v)|=k-1$.

\medskip

\noindent (ii) Take any $v \in \mathcal{V}$ and $i,j \in N$.
By construction, any $k \neq i,j$ is a null player in $R^{F,\{i,j\}}(\varphi(v),v)$.
Hence, by (F-NGC) and (EG), we have 
\begin{align*}
\varphi_i(v)-\varphi_j(v)
&=\varphi_i\left(R^{\text{F},\{i,j\}}(\varphi(v),v)\right)-\varphi_j\left(R^{\text{F},\{i,j\}}(\varphi(v),v)\right)\\
&=R^{\text{F},\{i,j\}}(\varphi(v),v)(\{i\})-R^{\text{F},\{i,j\}}(\varphi(v),v)(\{j\})\\
&=v(\{i\})-v(\{j\})\\
&=CIS_i(v)-CIS_j(v).
\end{align*}
Since $i,j \in N$ are arbitrarily, by (E), we conclude that $\varphi_i(v)=CIS_i(v)$.

\medskip

\noindent (iii) Take any $v \in \mathcal{V}$ and $i,j \in N$.
By construction, any $k \neq i,j$ is a null player in $R^{\text{M},\{i,j\}}(\varphi(v),v)$.
Hence, by (M-NGC) and (EG), we have 
\begin{align*}
\varphi_i(v)-\varphi_j(v)
&=\varphi_i\left(R^{\text{M},\{i,j\}}(\varphi(v),v)\right)-\varphi_j\left(R^{\text{M},\{i,j\}}(\varphi(v),v)\right)\\
&=R^{\text{M},\{i,j\}}(\varphi(v),v)(\{i\})-R^{\text{M},\{i,j\}}(\varphi(v),v)(\{j\})\\
&=v(N\setminus \{j\})-v(N \setminus \{i\})\\
&=ENSC_i(v)-ENSC_j(v).
\end{align*}
Since $i,j \in N$ are arbitrarily, by (E), we conclude that $\varphi_i(v)=ENSC_i(v)$.

\end{proof}

\section{Another Characterization Based on Active-player Consistency}\label{appendix:discussion_ac}

In Theorem \ref{consistency ELS}, we examined the implications of (AC) within the class of ELS values. However, the full strength of (L) is not needed for that characterization. In fact, the affine combinations of $(\psi^s)_{s=1}^{n-1}$ can also be characterized by a weaker axiom. Pick two players $i$ and $j$. For any group of the remaining players, we can ask how much that group gains from working with $i$ instead of $j$. Our next axiom, \emph{two-person linear bargaining} (TLB), considers the situation where only $i$ and $j$ are treated as active and all other players are kept at the payoffs they are already assigned by the solution in the original game. (TLB) requires that, in this situation, the difference between the payoffs of $i$ and $j$ is determined in a linear way by those comparisons: each comparison between $i$'s and $j$'s contributions is given a fixed weight, and the weighted sum pins down $i$'s payoff minus $j$'s payoff.

\medskip

\noindent\textbf{Two-Person Linear Bargaining (TLB)}: For any $v \in \mathcal{V}$ and $i,j \in N$,
$$
\varphi_i\left(R^{\text{AC},\{i,j\}}(\varphi(v),v)\right) - \varphi_j\left(R^{\text{AC},\{i,j\}}(\varphi(v),v)\right) = \gamma^{i,j}\left(\left(v(S\cup \{i\})-v(S \cup \{j\})\right)_{S \subseteq N\setminus \{i,j\}}\right),
$$
where $\gamma^{i,j}: \mathbb{R}^{2^{n-2}} \to \mathbb{R}$ is a linear function.

\medskip

(TLB) can be viewed as a rule for determining the payoffs between two players once the rest of the players have effectively been settled. Suppose only players $i$ and $j$ remain to be compared, and all other players are held at the payoffs already assigned to them by the solution. In that case, (TLB) requires that the difference between $i$'s and $j$'s payoffs be determined in a disciplined way: it must be a fixed linear summary of how much more valuable it is to include $i$ rather than $j$ in any possible partnership with the others. Thus, the solution treats the relative payoff of $i$ versus $j$ as an assigned payoff differential based purely on their comparative contributions, with no additional bargaining terms.

The special cases of the solution satisfying this property as well as (E) and (SYM) are illustrated as follows:

\vspace{-0.3cm} 

\begin{itemize}\itemsep-0.1cm
\item Suppose that for any $v \in \mathcal{V}$ and $i,j \in N$, 
$$
\varphi_i\left(R^{\text{AC},\{i,j\}}(\varphi(v),v)\right) - \varphi_j\left(R^{\text{AC},\{i,j\}}(\varphi(v),v)\right)=v(\{i\})-v(\{j\}).
$$
Then, $\varphi = CIS$.

\item Suppose that for any $v \in \mathcal{V}$ and $i,j \in N$, 
$$\varphi_i\left(R^{\text{AC},\{i,j\}}(\varphi(v),v)\right) - \varphi_j\left(R^{\text{AC},\{i,j\}}(\varphi(v),v)\right)=v(N\setminus \{j\})-v(N\setminus \{i\}).
$$
Then, $\varphi = ENSC$. 
\item Suppose that for any $v \in \mathcal{V}$ and $i,j \in N$, 
$$
\varphi_i\left(R^{\text{AC},\{i,j\}}(\varphi(v),v)\right) - \varphi_j\left(R^{\text{AC},\{i,j\}}(\varphi(v),v)\right)=\sum_{T \subseteq N\setminus \{i,j\}}\frac{|T|!(n-|T|-1)!}{n!}(v(T\cup \{i\})-v(T \cup \{j\})).
$$
Then, $\varphi = Sh$. 
\end{itemize}

\begin{lemma}\label{lem:etlbac_linear}
The following implications hold. 
\begin{itemize}\itemsep-0.1cm
\item[(i)] (L) and (SYM) together imply (TLB). 
\item[(ii)] (E), (TLB), and (AC) together imply (L). 
\end{itemize}
\end{lemma}

\begin{proof}
By Lemma \ref{Linear} (ii), it is trivial to verify the first implication. 
To show the second implication, let $v \in \mathcal{V}$. By (TLB) and (AC), for any $i,j \in N$, 
\[
\varphi_i(v)-\varphi_j(v)=\varphi_i\left(R^{\text{AC},\{i,j\}}(\varphi(v),v)\right)-\varphi_j\left(R^{\text{AC},\{i,j\}}(\varphi(v),v)\right)=\gamma^{i,j}\left( (v(S\cup \{i\})-v(S \cup \{j\})_{S \subseteq N\setminus \{i,j\}}  \right). 
\]
Moreover, by (E), $\sum_{i \in N}\varphi_i(v)=v(N)$.
By these equations, for any $i \in N$, $\varphi_i(v)$ is uniquely determined as
\[
\varphi_i(v)=\frac{1}{n} \left(v(N)-\sum_{j \neq i} \gamma^{j,i}\left( (v(S\cup \{j\})-v(S \cup \{i\}))_{S \subseteq N\setminus \{i,j\}}  \right)  \right).
\]
Since $\gamma^{j,i}$ is a linear function for any $i,j \in N$, so is $\varphi$.
\end{proof}

The second implication of Lemma~\ref{lem:etlbac_linear} fails without (TLB). 
To see this, consider the example in Appendix~\ref{Appendix: AC}, which satisfies (E) and (AC), but violates (TLB) and (L). 
Together with Theorem~\ref{consistency ELS}, the above argument characterizes the subclass of ELS values satisfying \textup{(AC)}.

\begin{corollary}\label{consistency ELS2}
Suppose that $\varphi$ satisfies (E), (TLB), and (SYM). Then the following are equivalent:

\vspace{-0.3cm}

\begin{itemize}\itemsep-0.1cm
\item[(I)] $\varphi$ satisfies (AC).
\item[(II)] $\varphi$ is an affine combination of $(\psi^s)_{s=1}^{n-1}$.
\end{itemize}
\end{corollary}

\section{Independence of the axioms}\label{appendix:independence}

\subsection{
(CU), (CD$_{\text{I}}$), and (CD$_{\text{O}}$)}

We provide several examples illustrating the independence of the axioms (CU), (CD$_{\text{I}}$), and (CD$_{\text{O}}$).

\begin{example}
\label{ex:CU}
Let
\begin{align*}
A^{1}(v) &:= \left\{ i \in N \mid v(\{i\}) \geq v(\{j\}) \ \forall j \in N \right\}, \\
A^{2}(v) &:= \left\{ i \in N \backslash A^{1}(v) \mid v(\{i\}) \geq v(\{j\}) \ \forall j \in N \backslash A^{1}(v) \right\},
\end{align*}
denote the sets of players with the highest and second-highest stand-alone worth, respectively. Define
\[
\varphi_k(v)=
\begin{cases}
\dfrac{v(N)}{n} + \dfrac{v(\{i\}) - v(\{j\})}{2} 
& \text{if } A^{1}(v)=\{i\},\ A^{2}(v)=\{j\},\ k=i, \\[4pt]
\dfrac{v(N)}{n} - \dfrac{v(\{i\}) - v(\{j\})}{2} 
& \text{if } A^{1}(v)=\{i\},\ A^{2}(v)=\{j\},\ k=j, \\[4pt]
\dfrac{v(N)}{n} 
& \text{otherwise.}
\end{cases}
\]
We show that $\varphi$ satisfies (CU), but neither (CD$_{\text{I}}$) nor (CD$_{\text{O}}$).

\medskip

\noindent (i) Suppose $A^{1}(v)=\{i\}$ and $A^{2}(v)=\{j\}$. By definition, $v(\{i\}) > v(\{j\})$. By construction of $\varphi$, we have
\[
v(\{i\}) - \varphi_i(v^{t}) = v(\{j\}) - \varphi_j(v^{t}) > v(\{k\}) - \varphi_k(v^{t}) \quad \text{for each } k \neq i,j.
\]
Hence, $A^{1}(U(\varphi(v^{t}),v)) = \{i,j\}$, and thus $\varphi(U(\varphi(v^{t}),v))$ is the egalitarian value.

On the other hand, again by construction of $\varphi$, the difference $\varphi(v) - \varphi(v^{t})$ is the egalitarian value, since the two allocations differ only in the grand coalition worth. Therefore, (CU) is satisfied.

\noindent (ii) In all other cases, $\varphi(v)$, $\varphi(v^{t})$, and $\varphi(U(\varphi(v^{t}),v))$ all coincide with the egalitarian value. Hence, (CU) is satisfied.

\medskip

We next show that $\varphi$ violates both (CD$_{\text{I}}$) and (CD$_{\text{O}}$). Suppose $A^{1}(v)=\{i\}$ and $A^{2}(v)=\{j\}$. By construction of $\varphi$, $\varphi_k(v^{t})$ attains its highest value when $k=i$, and its second-highest value when $k \neq i,j$. Hence, for both $D_{I}(\varphi(v^{t}),v)$ and $D_{O}(\varphi(v^{t}),v)$, the stand-alone worth is maximized at $k=i$, and is second-highest for all $k \neq i,j$. Therefore, both $\varphi(D_{I}(\varphi(v^{t}),v))$ and $\varphi(D_{O}(\varphi(v^{t}),v))$ are the egalitarian value. Since $\varphi(v)$ is not egalitarian, both (CD$_{\text{I}}$) and (CD$_{\text{O}}$) are violated.
\end{example}

\begin{example}
\label{ex:CDI}
We provide a solution $\varphi$ that satisfies (CD$_{\text{I}}$), but neither (CU) nor (CD$_{\text{O}}$). Let $A_{+}(v) := \{ i \in N \mid v(\{i\}) > 0 \}$ denote the set of players whose stand-alone worth is strictly positive. Let
\begin{equation}
\label{eq:def_phi_CDI}
\varphi_i(v)=
\begin{cases}
|v(N)| + 1 
& \text{if } A_{+}(v)=\{j\} \text{ and } i=j, \\[4pt]
\dfrac{v(N) - |v(N)| - 1}{n-1} 
& \text{if } A_{+}(v)=\{j\} \text{ and } i \neq j, \\[6pt]
\dfrac{v(N)}{n} 
& \text{otherwise.}
\end{cases}
\end{equation}

\medskip

We show that (CD$_{\text{I}}$) is satisfied.

\medskip

\noindent (i) Suppose $A_{+}(v)=\{j\}$. Then, for any $t \in \mathbb{R}$,
\begin{equation}
\varphi_i(v^{t})=
\begin{cases}
|t| + 1 > 0 & \text{if } i=j, \\[4pt]
\frac{t - |t| - 1}{n-1} < 0 & \text{if } i \neq j.
\end{cases}
\label{eq:CDI_phi_i}
\end{equation}
By definition of $D_{I}$, $D_{I}(\varphi(v^{t}),v)(\{i\})=\varphi_i(v^{t})$ for each $i \in N$, and thus $A_{+}(D_{I})=\{j\}$.\footnote{We refer to the reduced game $D_{I}(\varphi(v^{t}),v)$ by $D_{I}$ for short when there is no risk of confusion. The same convention applies to $D_{O}$ and $U$ below.}
Hence, $\varphi(D_{I})=\varphi(v)$, by construction of $\varphi$ given in (\ref{eq:def_phi_CDI}). (CD$_{\text{I}}$) is satisfied.

\medskip

\noindent (ii) In all other cases, both $\varphi(v)$ and $\varphi(v^{t})$ are egalitarian. Hence, $|A_{+}(D_{I})|$ is either $0$ (if $t \le 0$) or $n$ (if $t > 0$). Thus, both $\varphi(v)$ and $\varphi(D_{I})$ are egalitarian. Therefore, (CD$_{\text{I}}$) is satisfied.

\medskip

However, $\varphi$ does not satisfy (CD$_{\text{O}}$). To see this, consider a game $v$ such that $v(N)>0$, $v(\{j\})=1$, and $v(\{i\})=0$ for each $i \neq j$. Then $A_{+}(v)=\{j\}$, and (\ref{eq:CDI_phi_i}) applies. By definition of $D_{O}$,
$
D_{O}(\{i\}) = v(N) - t + \varphi_i(v^{t})$ \text{for each } $i \in N$.
By choosing $v(N)-t$ sufficiently large (e.g., $v(N)=100$ and $t=1$), we obtain $D_{O}(\{i\})>0$ for all $i \in N$. Hence, $A_{+}(D_{O})=N$, and thus $\varphi(D_{O})$ is egalitarian, whereas $\varphi_j(v)=|v(N)|+1$. Therefore, (CD$_{\text{O}}$) is violated.

\medskip

Moreover, $\varphi$ does not satisfy (CU). To see this, consider a game $v$ such that $v(N)=v(\{j\})>2$ and $v(\{i\})=0$ for each $i \neq j$. Then $A_{+}(v)=\{j\}$, and (\ref{eq:CDI_phi_i}) applies. Take $t \in (0,1)$. Then,
$
U(\varphi(v^{t}),v)(\{j\}) = v(\{j\}) - \varphi_j(v^{t}) 
= v(\{j\}) - (|t|+1) > 0,
$
and for each $i \neq j$,
$
U(\{i\}) = v(\{i\}) - \varphi_i(v^{t}) = -\varphi_i(v^{t}) > 0.
$
Hence, $A_{+}(U)=N$, and thus $\varphi(U)$ is egalitarian. However,
$
\varphi_j(v)=|v(N)|+1$,
$
\varphi_j(v^{t})=|t|+1,$
and their difference does not coincide with the egalitarian value $\varphi_j(U)=\frac{v(N)-t}{n}$. Therefore, (CU) is violated.
\end{example}

\begin{example}
\label{ex:bothCD}Let 
\[
\varphi _{i}(v)=\left\{ 
\begin{array}{ll}
\frac{v(N)}{n}+n-1 & \text{ if }A^{1}\left( v\right) =\left\{ j\right\} 
\text{ and }i=j, \\ 
\frac{v(N)}{n}-1 & \text{ if }A^{1}\left( v\right) =\left\{ j\right\} \text{
and }i\neq j, \\ 
\frac{v(N)}{n} & \text{ otherwise.}%
\end{array}%
\right. 
\]%

To see that both (CD$_{\text{I}}$) and (CD$_{\text{O}}$) are satisfied, suppose $A^{1}(v)=\{j\}$. Then player $j$ attains the maximal stand-alone worth in both reduced games $D_{I}$ and $D_{O}$, as well as in $v$. If $|A^{1}(v)|\neq 1$, the allocation is egalitarian in $D_{I}$, $D_{O}$, and $v$. In either case, both (CD$_{\text{I}}$) and (CD$_{\text{O}}$) are satisfied.

To see that (CU) is violated, consider a game $v$ such that $v(\{1\})>n$ and $v(\{i\})=0$ for each $i\neq 1$. Then $A^{1}(v)=A^{1}(v^{t})=\{1\}$ and also $A^{1}(U)=\{1\}$. Hence, player $1$ receives an allocation exceeding that of the other players by $n$ in each of $\varphi(v)$, $\varphi(v^{t})$, and $\varphi(U)$, violating (CU).
\end{example}

\begin{example}
\label{ex:onlyCDO}
This example satisfies (CD$_{\text{O}}$), but violates (CU) and (CD$_{\text{I}}$). It also satisfies (SYM), but violates (E) and (L). Assume $\overline{v}>\underline{v}\geq 0$ and restrict attention to the class of games with $v(N)\in (\underline{v},\overline{v})$.\footnote{We did not find such an example on the full domain $\mathcal{V}$. Whether (CD$_{\text{O}}$) implies (CD$_{\text{I}}$) on $\mathcal{V}$ remains an open question.}

Define $\varphi$ by
\[
\varphi_i(v)=
\begin{cases}
\overline{v} & \text{if } A^{1}(v)=\{j\},\ v(\{j\})\ge \underline{v},\ i=j,\\
0 & \text{if } A^{1}(v)=\{j\},\ v(\{j\})\ge \underline{v},\ i\ne j,\\
\overline{v}+1 & \text{if } A^{1}(v)=\{j\},\ v(\{j\})< \underline{v},\ i=j,\\
\overline{v} & \text{if } A^{1}(v)=\{j\},\ v(\{j\})< \underline{v},\ i\ne j,\\
\dfrac{v(N)}{n} & \text{otherwise.}
\end{cases}
\]

{We first show that (CD$_{\text{O}}$) holds.}

(i) If $A^{1}(v)=\{j\}$ and $v(\{j\})\ge \underline{v}$, then for any $t\in(\underline{v},\overline{v})$,
\[
D_O(\varphi(v^t),v)(\{i\})=
\begin{cases}
v(N) & i=j,\\
v(N)-\overline{v} & i\ne j.
\end{cases}
\]
Hence $A^{1}(D_O)=\{j\}$ with $D_O(\{j\})\ge \underline{v}$, and $\varphi(D_O)=\varphi(v)$.

(ii) If $A^{1}(v)=\{j\}$ and $v(\{j\})<\underline{v}$, then
\[
D_O(\varphi(v^t),v)(\{i\})=
\begin{cases}
v(N)-(n-1)\overline{v} & i=j,\\
v(N)-(n-1)\overline{v}-1 & i\ne j.
\end{cases}
\]
Thus $A^{1}(D_O)=\{j\}$ and $D_O(\{j\})<\underline{v}$, implying again $\varphi(D_O)=\varphi(v)$.

(iii) Otherwise, $\varphi_i(v)=\varphi_i(v^t)=\frac{v(N)}{n}$ for all $i$, and
$
D_O(\varphi(v^t),v)(\{i\})=\frac{v(N)}{n}.
$
Hence $A^{1}(D_O)=N$ and $\varphi(D_O)=\varphi(v)$. 

We next show that (CD$_{\text{I}}$) fails. In case (ii),
\[
D_I(\varphi(v^t),v)(\{i\})=\varphi_i(v^t)=
\begin{cases}
\overline{v}+1 & i=j,\\
\overline{v} & i\ne j.
\end{cases}
\]
Thus $A^{1}(D_I)=\{j\}$ with $D_I(\{j\})\ge \underline{v}$, so
\[
\varphi_i(D_I)=
\begin{cases}
\overline{v} & i=j,\\
0 & i\ne j,
\end{cases}
\]
which differs from $\varphi(v)$. 

{We show that (CU) fails.}
Let $A^{1}(v)=\{1\}$ with $v(\{1\})>\overline{v}$. Then $\varphi_1(v)=\varphi_1(v^t)=\overline{v}$. Also, we have $U(\{1\})=v(\{1\})-\overline{v}>0$ and $U(\{i\})=0$ for $i\ne 1$. Hence, $A^{1}(U)=\{1\}$, and thus $\varphi_1(U)>0$. So $
\varphi_1(v) < \varphi_1(v^t) + \varphi_1(U),$ violating (CU).
\end{example}

\begin{table}[H]
\centering
\begin{tabular}{c|ccc|ccc}
& (E) & (L) & (SYM) & (CU) & (CD$_\text{I}$) & (CD$_\text{O}$) \\ \hline
Zero allocation & $\times$ & $\checkmark$ & $\checkmark$ & $\checkmark$ & $\checkmark$ & $\checkmark$ \\ \hline
Example \ref{ex:CU} & $\checkmark$ & $\times$ & $\checkmark$ & $\checkmark$ & $\times$ & $\times$ \\
Example \ref{ex:CDI} & $\checkmark$ & $\times$ & $\checkmark$ & $\times$ & $\checkmark$ & $\times$ \\
Example \ref{ex:bothCD} & $\checkmark$ & $\times$ & $\checkmark$ & $\times$ & $\checkmark$ & $\checkmark$ \\
Example \ref{ex:onlyCDO}$^*$ & $\times$ & $\times$ & $\checkmark$ & $\times$ & $\times$ & $\checkmark$ \\ \hline
Dictatorship & $\checkmark$ & $\checkmark$ & $\times$ & $\checkmark$ & $\checkmark$ & $\checkmark$ 
\end{tabular}
\caption{Independence of the axioms. $^*$Full domain is also dropped in Example \ref{ex:onlyCDO}.}
\label{tab:indep_axioms}
\end{table}

Finally, we conclude this section by providing examples that satisfy the composition axioms after dropping either (E) or (SYM). The independence of the axioms illustrated by these examples is summarized in Table~\ref{tab:indep_axioms}.

\noindent\textbf{Zero allocation}: For each $v \in \mathcal{V}$ and each $i \in N$, $\varphi_i(v) = 0$. 

\noindent\textbf{Dictatorship}: There is $i \in N$ such that for each $v \in \mathcal{V}$, $\varphi_i(v) = v(N)$ and  $\varphi_j(v) = 0$ for each $j \in N \backslash \{i\}$. 

\subsection{(AC) and (L)}\label{Appendix: AC}
We provide an example of the solution that satisfies (E), (SYM), (AC), but violates (L).
For each $v \in \mathcal{V}$, let 
\[
A^{1}(v) = \left\{ i \in N \mid v(\{i\}) \geq v(\{j\}) \ \forall j \in N \right\}.
\]
For any $v \in \mathcal{V}^N$, define the solution $\varphi$ as follows.
\begin{itemize}
\item[\rm{(i)}] If there exists $i \in N$ such that $A^1(v)=\{i\}$ and $v(N) \ge v(\{i\})>v(N)+\max_{j \neq i}v(\{j\})$, 
then
\[
\varphi_k(v)=
\begin{cases}
v(N) & \text{if } k=i,\\
0 & \text{otherwise}.
\end{cases}
\]

\item[\rm{(ii)}] Otherwise, $\varphi(v)=CIS(v)$.
\end{itemize}
Note that condition (i) is indeed satisfied for some games; for example, this is the case when $v(N)>0=v(\{i\})$ and $\max_{j \neq i}v(\{j\})$ is sufficiently small. 
It is easy to verify that the solution satisfies (E) and (SYM), but violates (L). 
Thus, we prove that it satisfies (AC) using the following two claims. 
Take any $v \in \mathcal{V}^N$.

\begin{claim}\label{ex: AC-claim 1}
If condition (i) applies to $v$, then condition (i) also applies to $R^S(v,\varphi(v))$ for any $S \subsetneq N$.
\end{claim}

\begin{proof}
Let $A^1(v)=\{i\}$. 
Then, by definition, $\varphi_i(v)=v(N)$ and $\varphi_j(v)=0$ for every $j \neq i$.

First, suppose that $i \in S$. 
In this case, $R^S(v,\varphi(v))(\{j\})=v(\{j\})$ for every $j \in N$, and $R^S(v,\varphi(v))(N)=v(N)$. 
Since $\varphi$ depends only on $(v(\{j\}))_{j \in N}$ and $v(N)$, condition (i) also applies to $R^S(v,\varphi(v))$.

Next, suppose that $i \notin S$. 
In this case, $R^S(v,\varphi(v))(\{i\})=v(\{i\})-v(N)$, $R^S(v,\varphi(v))(\{j\})=v(\{j\})$ for every $j \neq i$, and $R^S(v,\varphi(v))(N)=v(N)-v(N)=0$. 
Moreover, since $v(\{i\})>v(N)+\max_{j \neq i}v(\{j\})$, we have
\[
R^S(v,\varphi(v))(\{i\})
=
v(\{i\})-v(N)
>
\max_{j \neq i}v(\{j\})
=
\max_{j \neq i}R^S(v,\varphi(v))(\{j\}).
\]
Thus, $A^1(R^S(v,\varphi(v)))=\{i\}$. 
Since $R^S(v,\varphi(v))(N)=0$, the preceding inequality also implies that
\[
R^S(v,\varphi(v))(\{i\})
>
R^S(v,\varphi(v))(N)
+
\max_{j \neq i}R^S(v,\varphi(v))(\{j\}).
\]
Similarly, since $v(N) \ge v(\{i\})$, we have $R^S(v,\varphi(v))(N)=0 \ge v(\{i\})-v(N)=R^S(v,\varphi(v))(\{i\})$. 
Hence, condition (i) also applies to $R^S(v,\varphi(v))$.
\end{proof}

\begin{claim}\label{ex: AC-claim 2}
If condition (i) applies to $R^S(v,\varphi(v))$ for any $S \subsetneq N$, then condition (i) also applies to $v$.
\end{claim}

\begin{proof}
Seeking a contradiction, suppose that condition (i) does not apply to $v$. 
Then, by definition, $\varphi(v)=CIS(v)$. 
For any $S \subsetneq N$, note that
\[
R^S(v,\varphi(v))(\{j\})
=
v(\{j\})-CIS_j(v)
=
-\frac{1}{n}\left(v(N)-\sum_{k \in N}v(\{k\})\right)
\]
for every $j \notin S$. 
Hence, for $A^1(R^S(v,\varphi(v)))=\{i\}$ to hold, it is either $S=\{i\}$ or $S=N\setminus\{i\}$.

First, suppose that $S=\{i\}$. 
Then $R^S(v,\varphi(v))(\{i\})=v(\{i\})$ and 
$R^S(v,\varphi(v))(N)=v(N)-\sum_{j \in N\setminus\{i\}}CIS_j(v)=CIS_i(v)$. 
Since condition (i) applies to $R^S(v,\varphi(v))$, we have
\begin{align*}
& R^S(v,\varphi(v))(\{i\})
>
R^S(v,\varphi(v))(N)
+
\max_{j \neq i}R^S(v,\varphi(v))(\{j\}) \\
&\Leftrightarrow \quad
v(\{i\})
>
CIS_i(v)
-
\frac{1}{n}\left(v(N)-\sum_{k \in N}v(\{k\})\right) \\
&\Leftrightarrow \quad
v(\{i\})>v(\{i\}),
\end{align*}
which is a contradiction.

Thus, we consider the case $S=N\setminus\{i\}$. 
By definition, $R^S(v,\varphi(v))(\{i\})=-(1/n)(v(N)-\sum_{k \in N}v(\{k\}))$, 
$R^S(v,\varphi(v))(\{j\})=v(\{j\})$ for every $j \neq i$, and 
$R^S(v,\varphi(v))(N)=v(N)-CIS_i(v)$. 
Since condition (i) applies to $R^S(v,\varphi(v))$, we have
\begin{align*}
& R^S(v,\varphi(v))(N) \ge R^S(v,\varphi(v))(\{i\}) \\
&\Leftrightarrow \quad
v(N)-CIS_i(v)
\ge
-\frac{1}{n}\left(v(N)-\sum_{k \in N}v(\{k\})\right) \\
&\Leftrightarrow \quad
v(N)\ge v(\{i\}).
\end{align*}
Similarly,
\begin{align*}
& R^S(v,\varphi(v))(\{i\})
>
R^S(v,\varphi(v))(N)
+
\max_{j \neq i}R^S(v,\varphi(v))(\{j\}) \\
&\Leftrightarrow \quad
-\frac{1}{n}\left(v(N)-\sum_{k \in N}v(\{k\})\right)
>
v(N)-CIS_i(v)+\max_{j \neq i}v(\{j\}) \\
&\Leftrightarrow \quad
v(\{i\})
>
v(N)+\max_{j \neq i}v(\{j\}).
\end{align*}
Together with $v(N)\ge v(\{i\})$, this inequality implies that 
$v(\{i\})>\max_{j \neq i}v(\{j\})$, and hence $A^1(v)=\{i\}$. 
Therefore, condition (i) applies to $v$, which is the desired contradiction.
\end{proof}

By Claims~\ref{ex: AC-claim 1} and~\ref{ex: AC-claim 2}, condition (i) applies to $v$ if and only if it also applies to $R^S(v,\varphi(v))$ for any $S \subsetneq N$. 
By the construction of $\varphi$ and Theorem~\ref{consistency ELS}, it follows that $\varphi$ satisfies (AC).

\end{appendix}

\bibliography{references.bib}

\end{document}